\documentclass[11pt]{article}
\usepackage{amsmath}
\usepackage{amssymb}
\usepackage{amsfonts}
\usepackage[english]{babel}
\usepackage{amsmath,amssymb,amsbsy,amsfonts,latexsym,amsthm,mathrsfs}
\usepackage{xcolor}
\usepackage[title]{appendix}
\usepackage{natbib}
\usepackage{placeins}
\numberwithin{equation}{section}
\usepackage{todonotes}
\usepackage{authblk}

\usepackage{amsmath,amsfonts,latexsym}
\newtheorem{theorem}{Theorem}[section]

\newtheorem{definition}[theorem]{Definition}
\newtheorem{lemma}[theorem]{Lemma}
\newtheorem{corollary}[theorem]{Corollary}
\newtheorem{remark}[theorem]{Remark}

\def\cl#1{{\mathscr #1}}

\def\P{{\mathbb P}}
\def\R{{\mathbb R}}

\def\N{{\mathbb  N}}

\def\E{{\mathbb E}}

\def\<{\langle}
\def\>{\rangle}

\def\eqlaw{\stackrel{\mathrm{law}}{=}}

\newcommand{\cvd}{\hfill \ensuremath{\Box}\medskip}

\def\ep{\varepsilon}

\begin{document}

\title{Short-time asymptotics for non self-similar \\
stochastic volatility models}

\author[1]{Giacomo Giorgio\thanks{giorgio@mat.uniroma2.it}}
\author[1]{Barbara Pacchiarotti\thanks{pacchiar@mat.uniroma2.it}}
\author[2]{Paolo Pigato\thanks{paolo.pigato@uniroma2.it}}
\affil[1]{Department of Mathematics, University of Rome Tor Vergata, Roma, Italy.}
\affil[2]{Department of Economics and Finance, University of Rome Tor Vergata, Roma, Italy.}

\date{\today }

\maketitle

\begin{abstract}
\noindent
We provide a short-time large deviation principle (LDP) for stochastic volatility models, where the volatility is expressed as a function of a Volterra process. This LDP does not require strict self-similarity assumptions on the Volterra process. For this reason, we are able to apply such an LDP to two notable examples of non self-similar rough volatility models: models where the volatility is given as a function of a log-modulated fractional Brownian motion [Bayer et al., Log-modulated rough stochastic volatility models. SIAM J. Financ. Math, 2021, 12(3), 1257-1284], and models where it is given as a function of a fractional Ornstein-Uhlenbeck (fOU) process [Gatheral et al., Volatility is rough. Quant. Finance, 2018, 18(6), 933-949]. In both cases we derive consequences for short-maturity European option prices, implied volatility surfaces and implied volatility skew. In the fOU case we also discuss moderate deviations pricing and simulation results.
\end{abstract}

\noindent%

\noindent%

\noindent%

\medskip

{\bf Acknowledgements:} We are grateful to Christian Bayer, Lucia Caramellino for discussion and support.
Barbara Pacchiarotti acknowledges the support of
Indam-GNAMPA (research project \lq\lq Stime asintotiche: principi di
invarianza e grandi deviazioni\rq\rq), of  the MIUR Excellence
Department Project awarded to the Department of Mathematics,
University of Rome Tor Vergata (CUP E83C18000100006) and of
University of Rome Tor Vergata (research program \lq\lq Beyond Borders\rq\rq,
project \lq\lq Asymptotic Methods in Probability\rq\rq (CUP E89C20000680005)). 

\section{Introduction}

Recent years have seen wide interest in volatility modelling with Volterra processes in the quantitative finance community. This has been spurred by the success of rough volatility models, where volatility is a non-Markovian, fractional process \citep{GJR18}. In many instances, in order to produce this type of dynamics, volatility is expressed as a function of a Volterra process, i.e. a suitable deterministic kernel integrated against a Brownian motion. In this context a very useful feature of such kernel and of the corresponding fractional process is self-similarity. 

When looking at approximation formulas and asymptotics, self-similarity is usually key 
as it enables the translation of a small-noise result into a short-time one through space-time rescaling. This can then be used to price short maturity options (see the discussion at the end of Section 3 in \cite{Gu2}, and \cite{Gu1}). Based on this procedure, several short-time formulas are available for rough volatility models, if the volatility process is expressed as a function of a fractional Brownian motion (fBM) \citep{FZ17}, as a function of a Riemann-Liouville process (RLp), as in the rough Bergomi model \citep{BFGHS,FGP22,fukasawa2020}, or as a solution to a fractional SDE, as in the fractional Heston model \citep{forde_heston_H0}.

However, obtaining short-time approximation formulas is more difficult  if volatility depends on a process which is not self-similar, such as the fractional Ornstein-Uhlenbeck (fOU) process \citep{horvath_jacquier_lacombe_2019,GS17,GS18,GS20hedging,GS19,GS20}, or the log-modulated fBM (log-fBM) \citep{BHP20}. In this paper, we address this issue, providing a short-time large deviation principle (LDP) for Volterra-driven stochastic volatilities, where the usual self similarity assumption is replaced by a weaker scaling property for the kernel, that needs to hold only in asymptotic sense (see conditions {\bf(K1)} and {\bf (K2)} below). 
We prove this general result starting from \citep{CelPac},
where a pathwise LDP for the log-price was proved when the volatility is  function of a  family of Volterra processes, and the price is solution to a scaled differential equation. Here, under suitable short-time asymptotic assumptions on the Volterra kernel, we prove a short-time LDP for the log-price process.

With our general result, we analyse more in depth models with volatility given as a function of fOU or log-fBM, neither of which is self-similar. However, we note that both these processes can be seen as a perturbation of self-similar processes, so that our general result can be applied, assuming that the price process is a martingale and a moment condition on the price.

\smallskip

The first class of processes to which we apply our LDP are \emph{log-modulated rough stochastic volatility models}, introduced in \cite{BHP20} as a logarithmic perturbation of a more standard power-law Volterra stochastic volatility model, with volatility depending on a log-fBM. These models allow for the definition of a ``true'', continuous volatility process with roughness (Hurst) parameter $H\in[0,1]$ (including the ``super-rough'' case $H=0$), at the price of losing the self-similar structure of the power-law kernel. Differently from our LDP setting,  however, in \cite{BHP20} Edgeworth-type asymptotics are considered, meaning that log-moneyness is of the form $x \sqrt{t}$ ($t$ representing the time to maturity), while in order to observe a large deviations behavior we look here at a suitable log-moneyness regime (cf. equation \eqref{eq:log:moneyness}). This regime is consistent with Forde-Zhang LDP for rough volatility
\citep{FZ17} and the related large deviation results discussed below. 
When $H>0$, we obtain a short-time LDP for the log-price process and consequent short-time option pricing, implied volatility and implied skew asymptotics. For this class of processes  the rate function only depends on the self-similar power-law kernel, while the speed depends also on the modulating logarithmic function.
It is shown in \cite{BHP20} that when $H=0$ the implied volatility skew explodes as $t^{-1/2}$ (with a logarithmic correction), realising the model free bound in \cite{LEE}. Even though our proof only holds in the $H>0$ case, the expression we obtain for skew asymptotics, computed for $H=0$, is consistent with this model free bound. 
We note that \cite{BaPa} have recently proved that in the $H=0$ case, even if the log-modulated model is well defined, an LDP cannot hold.

\medskip

The second class of models to which we apply our LDP have a stochastic volatility given by a function of a \emph{fractional Ornstein-Uhlenbeck} process. We find that, in short-time, such model behaves exactly as the analogous model, with volatility process given as the same function, computed on a fBM (i.e., the model studied in \cite{FZ17}). More precisely, we mean that the two models satisfy an LDP with same speed and rate function.
It follows that also the short-time implied skew (computed as a suitable finite difference) is the same for fOU and fBM-driven stochastic volatility models.
For small time scales, fOU is, in a sense, close to fBM (see equation \eqref{eq:FOU}), even though fOU is not self-similar.
It is not uncommon when dealing with rough stochastic volatility, starting from the groundbreaking work \cite{GJR18}, to consider at times fBM, at times fOU, depending on which is most convenient for the problem at hand. Our result can be seen as a justification of this type of procedure, as it shows that pricing vanilla options with one or the other volatility does not matter (too much) for short maturities. Moreover, the fOU process is the most standard choice for a \emph{stationary process} with a fractional correlation structure. This is one of the reasons why it has been used as volatility process for option pricing and related issues\footnote{note that both fBM and RLp (as in Rough Bergomi) are non-stationary and give rise to non-stationary volatility processes}  \citep{horvath_jacquier_lacombe_2019,GS17,GS18,GS20hedging,GS19,GS20}. 

From our short-time LDP we formally derive the corresponding moderate deviations result, consistent with the one holding for self-similar rough volatility \citep{BFGHS}. We provide numerical evidence for both these large and moderate deviations results and for the skew asymptotics. We investigate on simulations how the choice between fOU and fBM dynamics in the volatility affects volatility smiles and skews, how accurate are our approximations, and how they depend on the mean reversion parameter.

\medskip

{\bf Background.} In recent years, rough volatility has been widely used in option pricing, due to the great fits it provides to observed volatility surfaces \citep{bayer2016pricing} and its ability to capture fundamental stylized facts of the implied volatility, notably the power-law explosion of the implied skew in short-time, which explodes as $t^{H-1/2}$ under rough volatility \citep{alos2007short,fukasawa2011asymptotic,fukasawa2017short}.  Many authors have argued that $H$ is actually positive but very close to $0$ \citep{bayer2016pricing,fukasawa2020}, which would give an extreme skew explosion close to $t^{-1/2}$. This is a model-free bound \citep{LEE}, that can be reached pricing options using ``singular'' local (or local-stochastic) volatility models \citep{pig19,fps2020}, but it is hard to obtain with (rough) stochastic volatility, as in the limit $H\to 0$ the volatility process usually degenerates and can be defined only as a distribution, not as a genuine process \citep{forde_bergomi_H0,forde_heston_H0,NR18,hager2020multiplicative}. Moreover, one observes a skew-flattening phenomenon, as $H\to 0$, in some of these models.
This was the main motivation, in \cite{BHP20}, to introduce the log-modulation of the power law kernel, allowing the corresponding stochastic volatility to be defined as a genuine process also in the $H\to 0$ limit. Technically, the logarithmic correction ensures that the variance remains finite even for $H=0$, that in turn avoids the subsequent definition problems as $H\to 0$, as well as the skew-flattening problem. We refer to \cite{BHP20} and references therein for a detailed discussion of the $H\to 0$ problem. 

Volatility was already taken as an exponential function of a fOU process with $H<1/2$ when rough volatility was first proposed in \cite{GJR18}.
On the one hand, mostly because of the desired self-similarity property of the volatility process, the exponential of a fBM or of a RLp\footnote{RLp is the stochastic process driving the volatility in the rough Bergomi model \citep{bayer2016pricing}}, are often used for pricing options. 
On the other hand, it is argued, e.g. in \cite{GJR18}, that taking a fOU with small mean reversion is not very different from taking a fBM in the volatility, with the considerable advantage that fOU is a stationary process \citep{Ch-ka-Ma}, while  fBM and RLp are not. For a thorough discussion on fOU driven volatilities and related implied volatilities we refer in particular to \cite{GS17,GS20}, for the relation of fOU to fast mean reverting Markov stochastic volatility we refer
to \cite{GS18,GS19}, for hedging under fOU volatility we refer to \cite{GS20hedging},
for portfolio optimization using fast mean reverting fOU process with $H>1/2$ we refer to \cite{FH18}.
 A small-noise LDP under fOU volatility, with other related results, has been proved in \cite{horvath_jacquier_lacombe_2019}, and LDP and moderate deviation principles for the rough Stein-Stein and other models, also in short-time, have been discussed in \cite{jacquier2020volterra} (see Remark \ref{rm:related} for details).

In this paper we consider short-time pricing asymptotics, i.e. pricing short maturity European options. This is a widely studied topic, as these short maturity pricing formulas provide methods for fast calibration, a quantitative understanding of the impact of 
model parameters on generated implied volatility surfaces, led to some widely used parametrizations of the volatility surface, and help in the choice of the most appropriate model to be fitted to data \citep{aitsahalia2019}.  Short maturity approximations are also used to obtain starting points for calibration procedures, which are then based on numerical evaluations. They have applications also to hedging, trading and risk management.

For notable results on short maturity valuation formulas under Markovian stochastic volatility we refer to \cite{osajima2015general}, and to \cite{MS03,MS07} for Markovian stochastic volatility with jumps.
Short-time skew and curvature under rough volatility have been discussed in \cite{fukasawa2017short,alos2017curvature}. Short maturity valuation formulas for European options and implied volatilities under rough stochastic volatility are given, e.g., in
\cite{FZ17,euch2018short,BFGHS,FGP21,FGP22,fukasawa2020}. Short maturity local volatility under rough volatility is studied in
\cite{BDFP2022}. 
Pathwise large and moderate deviation principles for rough stochastic volatility models are established in  \cite{horvath_jacquier_lacombe_2019, jacquier2017pathwise,jacquier2020volterra,Gu1,Gu2,
gulisashvili2020timeinhomogeneous,gulisashvili2022multi,GVZ,GVZ2, CelPac, CatPac}.

\medskip

{\bf Content of the paper.} We consider in Section \ref{sec:ldp:volterra} an LDP for stochastic volatility models with volatility driven by general Volterra processes. In particular, in Section \ref{sec:shorttime} we prove a short-time LDP for such models without relying on self-similarity. In Section \ref{sec:applications} we see how these results provide short-time LDPs in relevant, non self-similar examples such as the log-fBM and the fOU process. In Section \ref{sec:pricing} we derive practical consequences for option pricing and implied volatility, for volatility models where the volatility depends on log-fBM or fOU, at the large deviations regime. In the case of fOU, we also consider moderate deviations. A numerical study of the accuracy and dependence on relevant parameters of our results in the fOU case concludes the paper in Section \ref{sec:numerics}.

\section{Large deviations for Volterra stochastic volatility models}
\label{sec:ldp:volterra}

\subsection{Small-noise large deviations for the log-price}

We are interested in stochastic volatility models with asset price dynamics described by
\begin{equation}\label{eq:price-SDE}
dS_t =  S_t\,\sigma(V_t)d(\bar{\rho}\bar B_t+\rho B_t), \qquad 0\leq t \leq T,\\
\end{equation}
where  we set, without loss of generality, $S_0=1$ the initial price. The time horizon is $ T > 0 $,
 $ \bar B $ and $ B $ are two independent standard Brownian motions, $ \rho \in (-1,1) $ is a correlation coefficient and $ \bar{\rho}=\sqrt{1-\rho^2}$, so that $\widetilde{B}= \bar{\rho}\bar B+\rho B $ is a standard Brownian motion $ \rho $-correlated with $ B $.
We assume that the process $ V $ is a non-degenerate, continuous Volterra type Gaussian process of the form
\begin{equation}\label{eq:Volterra} V_t=
\int_0^t K(t,s)\, dB_s, \quad 0\leq t \leq T. \end{equation}
Here, the kernel  $ K $ is a  measurable and square integrable function on $ [0,T]^2 $, such that
$K(0,0)=0$, $ K(t,s)=0 $ for all $ 0\leq t < s \leq T $ and
$$\sup_{t\in[0,T]} \int_0^T K(t,s)^2 \,ds < \infty. $$
 One can verify that the covariance function of the process $ V $ defined as above is given by
$$k(t,s)= E[V_{t}V_{s}]=\int_0^{t\wedge s} K(t,u)K(s,u)\,du, \quad t,s \in [0,T].$$
We introduce now the modulus of continuity of the kernel $ K $, defined as
$$
M(\delta)=  \sup_{\{t_1,t_2 \in [0,T]: |t_1-t_2|\leq \delta\}}  \int_0^T |K(t_1,s)-K(t_2,s)|^2\,ds, \quad 0 \leq \delta \leq T.
$$
In order to ensure the continuity of the paths of $V$, we assume that $K$ satisfies the following condition.
\begin{enumerate}
\item[\bf(A1)\rm]
\label{ass:Volterra-process}
		There exist constants $ c > 0 $ and $ \vartheta > 0 $ such that $ M(\delta)\leq c\,\delta^\vartheta $ for all $ \delta \in [0,T] $.
\end{enumerate}		
Let us recall that the unique solution to equation \eqref{eq:price-SDE} is  $(e^{X_t})_{t\in[0,T]}$, where the log-price process is defined by
\begin{equation}\label{eq:log-price}
X_t=-\frac 12\int_0^t\sigma^2(V_s) ds +  \bar{\rho}\int_0^t\sigma(V_s) d\bar B_s + \rho\int_0^t \sigma(V_s) dB_s.
\end{equation}
\begin{definition}
A modulus of continuity is an increasing function  $\omega:[0,+\infty)\to[0,+\infty)$
such that $\omega(0) = 0 $ and $\lim_{x\to 0 }\omega (x)=0$.
 A function $f$  defined on $\R$  is
called locally $\omega$-continuous, if for every $\delta>0$ there exists a constant $R(\delta)>0$  such that for all
$x,y\in[-\delta, \delta]$, inequality  $|f(x)-f(y)|\leq R(\delta)\omega(|x-y|)$ holds.
\end{definition}
\begin{remark}\rm
For instance, if $\omega(x)=x^{\vartheta}$, $\vartheta\in(0,1)$,   the function $f$ is locally $\vartheta $-H\"{o}lder continuous.
If $\omega(x)=x$,    the function $f$ is locally Lipschitz continuous.
\end{remark}

We consider the following assumptions on the volatility function $\sigma$.

\begin{enumerate}
\item[($\bf\Sigma 1$)\rm] \label{ass:hp-sigma-II}\rm
$ \sigma: \R \longrightarrow (0,+\infty) $ is a locally $\omega$-continuous function for some modulus of continuity $\omega$.
\item[($\bf\Sigma 2$)\rm]
\label{ass:hp-sigma-III}\rm
		 There exist constants $\vartheta, M_1,M_2>0, $ such that
		$ \sigma(x)\leq M_1+ M_2 \,|x|^\vartheta, \quad  x\in \R.$
\end{enumerate}

From now on, we denote by  $C ([0, T])$ (respectively $ C_0([0,T])$) the set of
continuous functions on $[0, T]$ (respectively the set of
continuous functions on $[0, T]$ starting at $0$),  endowed with the topology induced by the sup-norm.

Let $ \gamma_.: \mathbb{N} \to \mathbb{R}_+ $ be an infinitesimal, decreasing function,  i.e.    $ \gamma_n \downarrow 0 $, as $ n \to +\infty$ .  For every $ n \in \mathbb{N} $, we  consider the following scaled version of  equation \eqref{eq:price-SDE}
$$\begin{array}{c}
\begin{cases}
dS_t^{n} =\gamma_n  S_t^{n}\sigma(V_t^n)d(\bar{\rho}\bar B_t+\rho B_t), \qquad 0\leq t \leq T,\\
S_0^{n}=1,
\end{cases}
\end{array}$$
The  log-price process $ X_t^{n}=\log S_t^{n}, $ $ 0\leq t \leq T,$
in the scaled model is
\begin{equation}\label{eq:log-price-scaled}
X_t^{n}=-\frac 12\gamma_n^{2}\int_0^t\sigma^2(V_s^n) ds +  \gamma_n\,\bar{\rho}\int_0^t\sigma(V_s^n) d\bar B_s +  \gamma_n\,\rho\int_0^t\sigma(V_s^n) dB_s.
\end{equation}
Here  the Brownian motion $ \bar{\rho}\bar B+\rho B $ is multiplied by a \it  small-noise \rm parameter $ \gamma_n$
and  the Volterra process $V^n$  is  of the form
$$V^n_t=\displaystyle\int_0^t K^n(t,s)dB_s \quad 0\leq t \leq T,
$$
where $K^n$ is a suitable kernel. It can be verified that the covariance function of the process $ V^n$, for every $ n \in \mathbb{N}, $ is given by
$$ k^n(t,s)=\displaystyle\int_0^{t\wedge s}K^n (t,u)K^n(s,u) \, du \quad \mbox{for } t,s \in [0,T].
$$
In the setting above, we are interested in an LDP for the family $( (\gamma_n B,V^n))_{n \in \mathbb{N}} $ (we recall basic facts and notations on LDP in Appendix \ref{app:ldp}).
Such an LDP holds under the following conditions on  the covariance functions, as seen in Theorem 7.4 in \cite{CelPac}.
\begin{enumerate}
\rm
\item[{\bf(K1)\rm}] \label{ass:cov-lim}	
There exist an infinitesimal function $ \gamma_n $ and a kernel {$ \widehat K $ }(regular enough to be the kernel of a continuous Volterra process) such that

\begin{equation} \label{eq:ker-limit}\displaystyle\lim_{n \to +\infty} \frac{K^n(t,s)}{{\gamma_n}}=\widehat K(t,s) \end{equation} and

	$$ 	 \displaystyle\lim_{n \to +\infty} \frac{\int_0^{T} K^n(t,u)K^n(s,u) du} {\gamma^2_n}
	=
	\int_0^{T}\widehat K(t,u)\widehat K(s,u) du
	$$
uniformly for $ t,s \in [0,T].$

\item[{\bf(K2)\rm}]\label{eq:ker-exp-tight} There exist constants $ \beta,M>0 $, such that, for every $ n\geq n_0$
$$\displaystyle\sup_{s,t \in [0,T], s\neq t} \frac{\int_0^{T}(K^n(t,u)-K^n(s,u))^2du}{\gamma^2_n|t-s|^{2\beta}}\leq M.
$$
\end{enumerate}

\begin{theorem}\label{th:LDP-B-hatB}
 Let $ \gamma: \mathbb{N} \to \mathbb{R}_+ $ be an infinitesimal function. Suppose Assumptions {\bf(K1)\rm} and {\bf(K2)\rm}  
 are fulfilled.
	Then $ ((\gamma_nB, V^n))_{n \in \mathbb{N}} $ satisfies an LDP on $C_0([0,T])^2$ with speed $ \gamma_n^{-2} $ and good rate function
	$$
		I_{(B,V)} (f,g)=
		 \begin{cases}
		\displaystyle \frac12  \int_0^T \dot{f}(s)^2 \, ds & (f,g) \in \cl H_{(B,V)}\\
		\displaystyle +\infty & (f,g)\notin \cl H_{(B,V)}
		\end{cases}
		$$
where
$$\cl H_{(B,V)}=\{(f,g) \in C_0([0,T])^2: f \in H_0^1[0,T], \, g(t)=\int_0^t \widehat K(t,u)\dot{f}(u)\,du, \quad 0\leq t \leq T\},$$
 where $\widehat K$ is defined in equation \eqref{eq:ker-limit} and $H_0^1=H_0^1[0,T]$ is the Cameron-Martin space.
\end{theorem}

If  Assumptions  {\bf($\bf\Sigma 1$)} and ($\bf\Sigma 2$) 
hold for the volatility function $\sigma(\cdot)$, 
we also have a sample path LDP  for the family  of  processes $ ((X_t^{n})_{t \in [0,T]})_{n \in \mathbb{N}}$ and for the family of random  variables
$ (X_T^{n})_{n \in \mathbb{N}}$ (see Section 7 in \cite{CelPac} for details).
Let us denote $\hat f(t)=\int_0^t \widehat K(t,u)\dot{f}(u)\,du$ for  $f \in H_0^1$.
\begin{theorem}\label{th:LDP-log-price}
Under  Assumptions {\bf($\bf\Sigma 1$)}, {\bf($\bf\Sigma 2$)}, {\bf(K1)\rm} and {\bf(K2)\rm},  we have:
i) the family  of  processes $ ((X_t^{n})_{t \in [0,T]})_{n \in \mathbb{N}}$
satisfies an LDP
with speed $\gamma_n^{-2}$ and good rate function
\begin{equation}\label{eq:rate-fun-infinite} I_X(x)=\begin{cases} \inf_{f\in H_0^1[0,T]}\left[\frac12 \lVert f\rVert_{H_0^1[0,T]}^2+\frac12\int_0^T \big(\frac{\dot{x}(t)-\rho \sigma(\hat{f}(t))\dot{f}(t)}{\bar{\rho}\sigma(\hat{f}(t))}\big)^2\,dt\right]& x \in H_0^1[0,T] \\
		\displaystyle +\infty & x \notin H_0^1[0,T];
		\end{cases}
		\end{equation}
	ii) the family  of  random variables $ (X_T^{n})_{n \in \mathbb{N}}$
satisfies an LDP
with speed $\gamma_n^{-2}$ and good rate function
\begin{equation}\label{eq:rate-fun-finite}I_{X_T}(y)=\inf_{f\in H_0^1[0,T]} \left[\frac12 \lVert f\rVert_{H_0^1[0,T]}^2+\frac12 \frac{\big(y-\int_0^T \rho \sigma(\hat{f}(t))\dot{f}(t)\,dt\big)^2}{\int_0^T \bar{\rho}^2\sigma^2(\hat{f}(t)) dt}\right], \quad y\in \R.\end{equation}
\end{theorem}
\begin{remark}\label{rem:I=J}\rm
From Theorem 4.8 in \cite{FZ17}, it follows that $ I_{X_T}(0)=0,$
$$
\begin{cases}
\inf_{y\geq x } I_{X_T}(y)=\inf_{y> x } I_{X_T}(y)=I_{X_T}(x)\quad \mbox{for}\,x>0,\\
\inf_{y\leq x } I_{X_T}(y)=\inf_{y< x } I_{X_T}(y)=I_{X_T}(x)\quad \mbox{for}\,x<0.
\end{cases}
$$
\end{remark}

\subsection{Short-time large deviations for the log-price}\label{sec:shorttime}
It is well known that if the Volterra process is self-similar we can pass  from small-noise to short-time regime (see the discussion at the end of Section 3 in \cite{Gu2}). However, in general this is not possible if the process is not self-similar.
In this section, we  obtain  a short-time LDP that does not rely on the self-similarity assumption,   by using the results of the previous section. 

Let $ \ep_.: \mathbb{N} \to \mathbb{R}_+ $ be a sequence decreasing to zero, i.e.   $\ep_n\downarrow 0$ as $n\to +\infty$.
For every $n\in\mathbb{N}$ and $t\in[0,T]$, if $V$ is a Volterra process as in (\ref{eq:Volterra}) we have
\begin{equation}\label{eq:short-Volterra}
V_{\ep_nt}=\int_0^{\ep_n t} K(\ep_nt,s) dB_s
\eqlaw
 \int_0^{t} \sqrt \ep_n  K(\ep_n t,\ep_n s) dB_s=\int_0^{t}  K^n(t,s) dB_s=V^n_t,\end{equation}
with $K^n(t,s)=\sqrt \ep_n  K(\ep_nt,\ep_n s)$.
Therefore for every $n\in\mathbb{N}$ and $t\in[0,T]$, if $X$ is as  in (\ref{eq:log-price}), we have
	\begin{eqnarray*}
		X_{\ep_nt}&=&-\frac 12\int_0^{\ep_nt}\sigma^2(V_s) ds +  \bar{\rho}\int_0^{\ep_nt}\sigma(V_s) d\bar B_s + \rho\int_0^{\ep_nt} \sigma(V_s) dB_s\\
&\eqlaw&-\ep_n\frac 12\int_0^{t}\sigma^2(V^n_s) ds +\sqrt\ep_n\,  \bar{\rho}\int_0^{t}\sigma(V^n_s) d\bar B_s +\sqrt\ep_n\, \rho \int_0^{t} \sigma(V^n_s) dB_s.
	\end{eqnarray*}
Define $V^n_t=V_{\ep_nt}$ and suppose the family of processes $ ( V^n)_{n \in \mathbb{N}} $ satisfies an LDP  with  speed $ \gamma_n^{-2} $ (depending on $\ep_n$). Suppose furthermore that the family 
 $ ((\gamma_nB, V^n))_{n \in \mathbb{N}} $ satisfies an LDP  with speed $ \gamma_n^{-2} $ (for details on this topic see Section 7 and in particular Theorem 7.4 in \cite{CelPac})
 and let $X^n$ be the process defined in (\ref{eq:log-price-scaled}).
	If we consider the processes,  defined on the same space, we have
	\begin{equation*}
		X^n_t-\gamma_n \ep_n^{-1/2}X_{\ep_nt}=\frac12(\gamma_n\ep_n^{1/2}-\gamma_n^2 )\int_0^{t}\sigma^2(V^n_s) ds.
			\end{equation*}

Let us recall that  two families  $ (Z^n)_{n\in\N} $ and $ (\tilde{Z}^n)_{n\in\N} $ of  random variables are \emph{exponentially equivalent} (at the speed $v_n$, with $v_{n}\to \infty$ as $n\to \infty$) if  for any $  \delta > 0 $,
	$$ \limsup_{n\to +\infty }\frac{1}{v_n} \log P(|\tilde{Z}^n-Z^n|>\delta) = -\infty. $$
As far as the LDP is concerned, exponentially equivalent families  are indistinguishable. See  Theorem 4.2.13 in \cite{DemZei}.

\begin{theorem}\label{th:small-time-LDP}
Under  Assumptions {\bf($\bf\Sigma 1$)}, {\bf($\bf\Sigma 2$)}, {\bf(K1)\rm} and {\bf(K2)\rm}, the two families $((X^n_t)_{t\in[0,T]})_{n\in\mathbb{N}}$ and $((\gamma_n \ep_n^{-1/2}X_{\ep_nt})_{t\in[0,T]})_{n\in\mathbb{N}}$   are exponentially equivalent  and therefore satisfy the same LDP. In particular,

(i)  the family $((\gamma_n \ep_n^{-1/2}X_{\ep_nt})_{t\in[0,T]})_{n\in\mathbb{N}}$ satisfies an LDP  with  speed $\gamma_n^{-2}$ and good rate function given by (\ref{eq:rate-fun-infinite});

(ii)  the family of random variables $ (\gamma_n \ep_n^{-1/2}X_{\ep_n T})_{n \in \mathbb{N}}$
satisfies an LDP  with  speed $\gamma_n^{-2}$ and good rate function given by (\ref{eq:rate-fun-finite}).
	\end{theorem}
\begin{proof}
We have
	\begin{equation*}
		|X^n_t-\gamma_n \ep_n^{-1/2}X_{\ep_nt}|=\frac12|\gamma_n\ep_n^{1/2}-\gamma_n^2 |\int_0^{t}\sigma^2(V^n_s)ds=\delta_n \int_0^{t}\sigma^2(V^n_s) ds,
			\end{equation*}
where $\delta_n=\frac12|\gamma_n\ep_n^{1/2}-\gamma_n^2 |$ and $\delta_n\to 0$.
The family $(V^n)_{n\in \N}$ satisfies an LDP with a good rate function. Then, it is exponentially tight (at the inverse speed $\gamma_n^2$). Therefore for every $R>0$, there exists a compact set $C_R$ (of equi-bounded functions) such that
$\limsup_{n\to +\infty} \gamma_n^2 \log \P\Big( V^n\in C^{c}_R\Big)\leq -R$, with $(\cdot)^{c}$ indicating the complementary set.
	Thus, for every $\eta>0$,
	$$\begin{array}{c}
	\limsup_{n\to+\infty}\gamma_n^2\log\mathbb{P}\Big( \sup_{t\in[0,T]}| X^n_t-\sqrt{\gamma_n}X_{\gamma_nt} |>\eta \Big) \phantom{ghdfasasasasfgdfhgfhdsaddsdsdsdh}\\
\leq\limsup_{n\to+\infty}\gamma_n^2\log\mathbb{P}\Big( \sup_{t\in[0,T]} \int_0^{t}\sigma^2(V^n_s) ds >\eta/\delta_n ,V^n\in C_R\Big)\\
 \phantom{ghdfasasasasfgdfhgfdsfdsfdsffdhds}+\limsup_{n\to+\infty}\gamma_n^2\log\mathbb{P}\Big(  V^n\in C_R^{c}\Big)=-\infty,
	\end{array}$$
since the set $\Big\{ \sup_{t\in[0,T]} \int_0^{t}\sigma^2(V^n_s) ds >\delta/\delta_n , V^n\in C_R\Big\}$ is eventually empty.
\cvd
\end{proof}


\section{Applications}\label{sec:applications}
In this section, we consider some (non self-similar) Volterra processes that satisfy 
assumption \bf (A1) \rm and such that the corresponding family $(V^n)_{n\in\N}$ defined by equation (\ref{eq:short-Volterra}) satisfies conditions {\bf(K1)\rm}  and {\bf(K2)\rm}.  
  We also suppose that assumptions ($\bf\Sigma 1$) \rm and ($\bf\Sigma 2$) \rm on the volatility function are satisfied and $T=1$.
From Theorem \ref{th:small-time-LDP} we obtain a short-time LDP for the corresponding  log-price processes.

\subsection{Log-fractional Brownian motion and modulated models}\label{sec:modulated}
Let us consider the kernel, for $0\leq s\leq t\leq 1$,
\begin{equation}\label{eq:kernel:logmod}
{K}
(t,s)=C (t-s)^{H-1/2} (-\log(t-s))^{-p},
\end{equation}
where $0\leq H \leq 1/2$,  $p>1$ and $C>0$ is a constant.
The corresponding Volterra process $V$ essentially amounts to
 the log-fBM introduced in \cite{BHP20}. There, an additional
 cutoff of the logarithm function was introduced in order to normalize the variance of the volatility at time one, but we can avoid here this complication as it does not affect our analysis in any way, since we only consider short-time asymptotics.

Condition \bf (A1) \rm for this kernel was proved in 
\cite{BHP20} with $\vartheta=2H$.
Note that $\kappa(t,s)=C (t-s)^{H-1/2}$ is the well known kernel of the RLp, which also satisfies Assumption \bf(A1) \rm with $\vartheta=2H$.

For $n$ large enough, we set
$$ K^n(t,s)=C  \ep_n^{H}(t-s)^{H-1/2} (-\log(\ep_n(t-s)))^{-p}.$$

Let us verify that  conditions {\bf(K1)\rm}  and {\bf(K2)\rm} are satisfied for $H>0$. No small time LDP can be verified in the case $H=0$, as shown in Section 5.4 in \cite {BaPa}. 

{\bf(K1)\rm}
For $s,t\in[0,1]$, $s<t$, since we can suppose $\ep_n<1$, we have
\begin{equation}\label{eq:ineq}\frac {\log \ep_n(t-s)}{\log \ep_n }\geq 1,\end{equation} and therefore
 $$\kappa(t,u)\Big(\frac{\log \ep_n(t-u) }{\log \ep_n}\Big)^{-p}\leq \kappa(t,u).$$ 
 Then, thanks to Lebesgue's dominated convergence Theorem, for  $s,t\in[0,1]$,
$$\lim_{n\to \infty}\int_0^{s\wedge t} \kappa(t,u)\kappa(s,u)\Big(\frac{\log \ep_n(t-u) }{\log \ep_n}\Big)^{-p}\Big(\frac{\log \ep_n(s-u) }{\log \ep_n}\Big)^{-p}du=\int_0^{s\wedge t}\kappa(t,u)\kappa(s,u)du,$$ 
so that
$\lim_{n\to \infty} \frac{k^n(t,s)}{\ep_n^{2H}(-\log \ep_n)^{-2p}}=k(t,s)$.

This convergence is actually uniform,  since
 $$\frac{k^n(t,s)}{\ep_n^{2H}(-\log \ep_n)^{-2p}}= C^2\int_0^{s\wedge t} (t-u)^{H-1/2}(s-u)^{H-1/2}\Big(1+ \frac {\log (t-u)}{\log \ep_n }\Big)^{-p}
\Big(1+ \frac {\log (s-u)}{\log \ep_n }\Big)^{-p} du,$$
and therefore the sequence $(k^n(\cdot,\cdot)/{\ep_n^{2H}(-\log \ep_n)^{-2p}})_n$ is a monotone sequence of continuous functions converging pointwise to a continuous function. Then  {\bf(K1)\rm}
is proved (with $\widehat K=\kappa$ and $\gamma_n=\ep_n^{H}(-\log \ep_n )^{-p}$). 
\medskip

{\bf(K2)\rm} For every $ n \in \mathbb{N}$, $s<t$,
we have
$$ \displaylines{\int_0^t \Big(\kappa(t,u) \frac{\log \ep_n(t-u) }{\log \ep_n}\Big)^{-p}-\kappa(s,u) \frac{\log \ep_n(s-u) }{\log \ep_n}\Big)^{-p}\Big)^2 \, du\cr
=\int_s^t \kappa(t,u)^2\Big(\frac{\log \ep_n(t-u) }{\log \ep_n}\Big)^{-2p}\, du+\int_0^s \Big(\kappa(t,u) \frac{\log \ep_n(t-u) }{\log \ep_n}\Big)^{-p}-\kappa(s,u)\frac{\log \ep_n(s-u) }{\log \ep_n}\Big)^{-p}\Big)^2\, du. }$$
Now, thanks to (\ref{eq:ineq})
$$
\int_s^t \kappa(t,u)^2\Big(\frac{\log \ep_n(t-u) }{\log \ep_n}\Big)^{-2p}\, du\\
\leq \int_s^t \kappa(t,u)^2\, du.$$

Let us prove that
$$\int_0^s \Big(\kappa(t,u) \Big(\frac{\log \ep_n(t-u) }{\log \ep_n}\Big)^{-p}-\kappa(s,u) \Big(\frac{\log \ep_n(s-u) }{\log \ep_n}\Big)^{-p}\Big)^2\, du\leq
\int_0^s (\kappa(t,u) -\kappa(s,u) )^2\, du.$$

The map $x\to x^{H-1/2}(-\log x)^{-p}$ defines a decreasing function in a neighbourhood of $x=0$  and $x\to (-\log x)^{-p}$  an increasing function for $x\in (0,1)$. Then,
for $n$ large enough, for $u\in(0,s)$, we have
\begin{eqnarray*}
0&\leq& (\ep_n u)^{H-1/2}(-{\log(\ep_nu)})^{-p}- (\ep_n(t-s+u))^{H-1/2} (-{\log(\ep_n(t-s+u))} )^{-p}\\
&\leq&   (\ep_n u)^{H-1/2}(-{\log(\ep_n(t-s+u))})^{-p}- (\ep_n(t-s+u))^{H-1/2} (-{\log(\ep_n(t-s+u))} )^{-p}\\ &=& 
( (\ep_nu)^{H-1/2}- (\ep_n(t-s+u))^{H-1/2}) (-{\log(\ep_n(t-s+u)\ep_n)} )^{-p}.
\end{eqnarray*}
Therefore,
\begin{eqnarray*}
&&\int_0^s \Big(\kappa(t,u) \Big(\frac{\log \ep_n(t-u) }{\log \ep_n}\Big)^{-p}
-\kappa(s,u) \Big(\frac{\log \ep_n(s-u) }{\log \ep_n}\Big)^{-p}\Big)^2\, du\\ 
&=&\!\!\!\frac {C^2 }{\ep_n^{2H}(-\log \ep_n)^{-2p}}\!\!\! \int_0^s \!\!\!\Big((\ep_n(t-s+u))^{H-\frac 12} (-{\log(\ep_n(t-s+u)} )^{-p}- (\ep_n\,u)^{H-\frac 12}  (-{\log(\ep_n u)})^{-p}\Big)^2\!\! du\\&\leq& \frac {C^2 }{\ep_n^{2H}(-\log \ep_n)^{-2p}}\int_0^s ( (\ep_nu)^{H-1/2}- (\ep_n(t-s+u))^{H-1/2})^2 (-{\log(\ep_n(t-s+u))} )^{-2p} du
\\& =& C^2 \int_0^s ( u^{H-1/2}- (t-s+u)^{H-1/2})^2 \Big(\frac{\log(\ep_n(t-s+u))}{\log \ep_n} \Big)^{-2p}\, du\\&\leq&
C^2 \int_0^s ( u^{H-1/2}- (t-s+u)^{H-1/2})^2\, du=\int_0^s (\kappa(t,u) -\kappa(s,u) )^2\, du,
\end{eqnarray*}
Therefore conditions  {\bf(K1)\rm} and {\bf(K2)\rm}  are verified with infinitesimal function $\gamma_n=\ep_n^H(-\log \ep_n)^{-p}$, limit kernel $\widehat K=\kappa$, and
 $\beta=H$.
A short-time LDP  holds with inverse speed $\ep_n^{2H}(-\log \ep_n )^{-2p}$  and limit kernel $\kappa(t,s)=C (t-s)^{H-1/2}.$ 

\medskip

The results proved for the log-fBM can be extended to a
class of processes, that we refer to as
  \emph{modulated Volterra processes}, defined,  for $t\in[0,1]$, as
\begin{equation}\label{eq:Volterra-modulated} V_t= \int_0^t \kappa(t,s) L(t-s) dB_s.\end{equation}  
Here, 
$\kappa$ is the kernel of a self-similar  Volterra process of index $H>0$,
i.e.
\begin{equation}\label{eq:self-similar}
\kappa( c t,c s)= c^{H-1/2} \kappa(t,s), \quad \mbox{ for }  c>0,
\end{equation}
that satisfies Assumption \bf(A1)\rm  \it, modulated  \rm by a slowly varying function $L$, i.e. a function such that
$$\lim_{x\to 0^+}\frac { L(x\lambda)}{L(x)}=1,$$
for every $\lambda>0$. 
Thanks to (\ref{eq:short-Volterra})  and (\ref{eq:self-similar}),   we have
$$V^n_t=\int_0^{t}  \ep_n^H  \kappa(t, s) L(\ep_n (t-s))dB_s.$$
First we note that here $K^n(t,s)=\ep_n^H  \kappa(t, s) L(\ep_n (t-s)) $ for $s,t\in[0,1]$, $s<t$ and
$$\lim_{n\to \infty} \frac{K^n(t,s)}{L(\ep_n)\ep_n^H}=\lim_{n\to \infty}\kappa(t,s) \frac{L(\ep_n (t-s))}{L(\ep_n)}=\kappa(t,s).$$

Note that the limit kernel is independent of $L$. For these processes, if assumptions  {\bf(K1)\rm}  and {\bf(K2)\rm} are satisfied,  we have a short-time LDP with the same rate function as the self-similar process and inverse speed
$L(\ep_n)^{2}\ep_n^{2H}$. Therefore, the rate function does not depend on the modulating function $L$, but the speed of the LDP does.

\subsection{Fractional Ornstein-Uhlenbeck process}\label{sec:fOU}

Let us recall that the Mandelbrot-Van Ness fBM $B^{H}$ is the 
centered continuous Gaussian process with covariance function
\[
\frac{1}{2}\left(t^{2H}+s^{2H}-|t-s|^{2H}\right).
\]
This process is self-similar with exponent
 $H$ and admits a Volterra representation with kernel (see e.g. \cite{Nualart:06})
\begin{equation} \label{eqn:kernelfbm}
K_H(t,s) = c_H \left[ \left( \frac{t}{s} \right) ^{H-1/2}(t-s)^{H-1/2} - \left( H-\frac{1}{2} \right) s^{1/2 - H} \int_s^t \!u^{H-3/2}(u-s)^{H-1/2} du \right],
\end{equation}
where
\[
c_H = \Big(  \frac{2H \, \Gamma(3/2-H)}{\Gamma(H+1/2) \, \Gamma(2-2H)}\Big)^{1/2}.
\]
For $H\in(0,1)$ and $a > 0$, we consider the fOU process, solution to
\[
 d V_{t} = -a V_{t } dt + dB^{H}_{t}, 
\]
which is given explicitly, with initial condition $V_{0}=0$, 
by\footnote{
Let us mention that, for other purposes, one could consider the stationary solution to the fractional SDE above (see for example \cite{GJR18}), explicitly given by 
$\int_{-\infty}^t e^{-a(t-u)} dB^H_u.$
However, we are interested in this paper in option valuation, so we take as volatility driver the process $V_t$ above, with $V_{0}=0$, so that $\sigma_{0}=\sigma(V_{0})=\sigma({0})$ is spot volatility in \eqref{eq:stoch:vol:fOU}. 
}
$$
V_t= \int_0^t e^{-a(t-u)} dB^H_u, \quad t\geq 0.
$$
Here, the stochastic integral with respect to $B^{H}$ can be defined, by integration
by parts and the stochastic Fubini theorem, as
\begin{equation}\label{eq:FOU}
V_t=B^H_t - a\int_0^t e^{-a(t-u)}B^H_u du.
\end{equation}

We note from this equation that self-similarity for fOU is approximately inherited from the fBM, for small time scales.
From \eqref{eq:FOU} we obtain, for $V$, the Volterra representation 
$$V_t=\int_0^t  K(t,s) dB_s,$$
with
$$
 K(t,s)= K_H(t,s)- a \int_s^t e^{-a(t-u)} K_H(u,s) du,$$
 and $K_{H}$ as above
(see, e.g., Section 2 in \cite{CelPac}). Condition \bf (A1)  \rm for this process, with $\vartheta=\min\{2H,1\}$, was established in Lemma 10 in \cite{Gu1}.
Here we have
$$  K^n(t,s)=  \ep_n^{H}K_H(t,s)- a \ep_n^{H+1} \int_s^t e^{-a\ep_n(t-u)} K_H(u,s) du.$$

Let us verify that  conditions {\bf(K1)\rm}  and {\bf(K2)\rm} are satisfied.

{\bf(K1)\rm}  
It is enough to observe that 
$$\left|\frac{ K^n(t,s)}{{\ep_n^H}}-K_H(t,s)\right|=a \ep_n\int_s^t e^{-a\ep_n(t-u)} K_H(u,s) du\leq C\ep_n, $$
where $C>0$ is a constant independent of $s,t\in[0,1]$. Therefore
$$\displaystyle\lim_{n \to +\infty} \frac{ K^n(t,s)}{{\ep_n^H}}=K_H(t,s),$$
uniformly for $s,t\in[0,1]$.
Therefore also  
$$ \int_0^{t\wedge s}\widehat K(t,u)\widehat K(s,u) du= \displaystyle\lim_{n \to +\infty} \frac{\int_0^{t\wedge s} K^n(t,u)K^n(s,u) du} {\gamma^2_n}$$
uniformly for $ t,s \in [0,T]$ and {\bf(K1)\rm} is proved (with $\widehat K=K_H$ and $\gamma_n=\ep_n^{H}$).

{\bf(K2)\rm}  For $s<t$ we have 
\begin{eqnarray*}&&\left| K^n(t,u)- K^n(s,u)\right|\\
&\leq &\ep_n^{H}|K_H(t,u)- K_H(s,u)|+ a\ep_n^{H+1}\left| 
\int_u^t e^{-a\ep_n(t-v)} K_H(v,u)dv -\int_u^s e^{-a\ep_n(s-v)} K_H(v,u)dv\right|\\
&\leq& \ep_n^{H}|K_H(t,u)- K_H(s,u)|+ a\ep_n^{H+1}\int_s^t e^{-a\ep_n(t-v)} K_H(v,u)dv \cr & +&a\ep_n^{H+1}\left| 
\int_u^s \Big(e^{-a\ep_n(t-v)}  - e^{-a\ep_n(s-v)} \Big)k_H(v,u)dv\right|\\ 
&=&\ep_n^{H}|K_H(t,u)- K_H(s,u)|+ a\ep_n^{H+1}e^{-a\ep_n(t-s)} \int_s^t e^{-a\ep_n(s-v)} K_H(v,u)dv \cr &+&a\ep_n^{H+1}|e^{-a\ep_n(t-s)}-1| 
\int_u^s e^{-a\ep_n(s-v)} k_H(v,u)dv.
\end{eqnarray*}

Therefore, denoting by $C>0$ a  constant (not depending on $s,t\in[0,1]$), we have
\begin{eqnarray*}
&&\int_0^{1}(K^n(t,u)-K^n(s,u))^2du\cr &\leq&
\ep_n^{2H}{\int_0^{1}(K_H(t,u)-K_H(s,u))^2du}+
\ep_n^{2H+2}\int_0^1 du \Big(\int_s^t e^{-a\ep_n(t-v)} K_H(v,u)dv\Big)^2\cr &+& C (t-s)^2 \ep_n^{2H+4}
\int_0^1\Big(\int_u^s \Big(e^{-a\ep_n(t-v)} K_H(t,v)\Big)^2\cr&\leq&
\ep_n^{2H}{\int_0^{1}(K_H(t,u)-K_H(s,u))^2du}+
\ep_n^{2H+2}(t-s)\int_0^1  \int_0^1  K_H(v,u)^2 du\,dv\cr &+& C (t-s)^2 \ep_n^{2H+4}
\int_0^1  \int_0^1  K_H(v,u)^2 du\,dv\leq C(t-s)^{2H\wedge 1}\ep_n^{2H},
\end{eqnarray*}
since (see for example Lemma 8 in \cite{Gu1})
$$\displaystyle\sup_{s,t \in [0,T], s\neq t} \frac{\int_0^{1}(K_H(t,u)-K_H(s,u))^2du}{|t-s|^{2H}}\leq M.
$$
Condition  {\bf(K2)\rm} is  verified with infinitesimal function $\gamma_n=\ep_n^H$,  $\beta=H$ and limit kernel $\widehat K=K_H$. 
Therefore, the short-time asymptotic behaviour of the model with volatility given as a function of the fOU process is exactly the same as the one of the model with volatility  given as a function of the fBM, meaning that they both satisfy LDPs where the speed and rate function are the same. Indeed, the rate function in \eqref{eq:rate-fun-finite} is the same that was found in \cite{FZ17}. This can be computed numerically as we do in Section \ref{sec:numerics}.

\section{Short-time asymptotic pricing and implied volatility}\label{sec:pricing}
In this section we discuss applications to option pricing and behaviour at short maturities of implied volatility for certain stochastic volatility models, using the LDP previously discussed.
We denote 
\begin{equation}\label{eq:putcall}
p(t,k):=E[(e^{k}-S_t)^+], \quad c(t,k):=E[(S_t-e^{k})^+],
\end{equation}
the European put and call prices with maturity $t$ and log-moneyness $k$ (i.e., strike $e^{k}$, since $S_{0}=1$).

\subsection{Large deviations pricing for log-modulated models}

Let us consider the stochastic volatility model given by \eqref{eq:price-SDE} and \eqref{eq:Volterra-modulated}, i.e.
\[
\begin{cases}
dS_t &=S_t\sigma(V_t)d(\bar{\rho}\bar B_t+\rho B_t),\\
V_t&= \int_0^t \kappa(t,s) L(t-s) dB_s,
\end{cases}
\]
with $\kappa$ kernel of a self-similar process, of exponent $H<1/2$, that satisfies {\bf(A1)\rm}, and $L$ slowly varying, such that $K(t,s)=\kappa(t,s) L(t-s)$ satisfies {\bf(A1)\rm}, {\bf(K1)\rm}, {\bf (K2)\rm}. In particular, this holds true for the kernel in \eqref{eq:kernel:logmod}, that essentially is the kernel of the log-fBM in \cite{BHP20}, for $H\in[0,1/2)$. Let 
\[
\Lambda(x)=I_{X_1}(x),
\]  
where $I_{X_1}$ is the rate function in \eqref{eq:rate-fun-finite}. 

Let us write $f_{t}\approx g_{t}$ if $\log(f_{t})\sim \log(g_{t})$ (see also Appendix \ref{app:ldp}).

\begin{theorem}\label{th:ivol:logmod}
Let us assume that {\bf(A1)\rm}, {\bf(K1)\rm}, {\bf (K2)\rm},  {\bf($\bf\Sigma 1$)}, {\bf($\bf\Sigma 2$)} hold.
If $x<0$ and
\begin{equation}\label{eq:log:moneyness}
k_t=xt^{-H+1/2}L(t)^{-1},
\end{equation}
the short-time put price satisfies
\[
p(t,k_t)=E[(e^{k_t}-S_t)^+]\approx \exp\{-t^{-2H} L(t)^{-2} \Lambda(x) \}.
\]
Let us now assume that the process $S$ is a martingale and there exist $p>1,t>0$ such that $E [ S_t^p ]<\infty$ (cf. Remark \ref{rm:mart}).
If $x>0$, $k_t$ is as in \eqref{eq:log:moneyness}, we have
\[
c(t,k_t)=E[(S_t-e^{k_t})^+]\approx \exp\{-t^{-2H} {L(t)^{-2} } \Lambda(x) \}.
\]
\end{theorem}
\begin{proof} 
We just prove the call asymptotics (the least straightforward).
From Theorem \ref{th:LDP-log-price} and Theorem \ref{th:small-time-LDP}, following the computations in Section \ref{sec:modulated}, we have that
the family $ (\ep_n^{H-1/2}L(\ep_n)X_{\ep_n })_{n \in \mathbb{N}}$ satisfies an LDP  with inverse speed $\ep_n^{2H}L(\ep_n)^2$ and good rate function  $I_{X_1}$ given by formula (\ref{eq:rate-fun-finite}).
Since $\inf_{y\geq x} I_{X_1}(y)=\inf_{y>x} I_{X_1}(y)=I_{X_1}(x)$  (see Remark \ref{rem:I=J})  we have for $x>0$
$$
 \displaylines{ \lim_{n \to +\infty}\ep_n^{2H}L(\ep_n)^2\log\P(\ep_n^{H-1/2}L(\ep_n)X_{\ep_n }> x)=-\Lambda(x),}
$$
for every sequence $\ep_n\downarrow 0$. 
Therefore, setting $\nu_t=t^{H-1/2}L(t)$, so that $k_t=x/\nu_t$, we have 
$$
 \displaylines{ \lim_{t\to 0}t^{2H}L(t)^2\log\P(X_{t }> k_t)= \lim_{t\to 0}t^{2H}L(t)^2\log\P(\nu_tX_{t }> x)=-\Lambda(x),}
$$
i.e.,
\begin{equation}\label{eq:ldp:st}
\P(S_t> e^{k_t})=\P(X_t>k_t)=\P(\nu_tX_{t }> x) \approx \exp(-t^{-2H}L( t)^{-2} \Lambda(x) ).
\end{equation}
Let us prove the upper bound. Let  $t>0$ be small enough such that  $\nu_t\geq 1$ and fix $y>x$. 
We have
\begin{eqnarray*}
E [ (S_t - \exp (k_t))^+]
&=&  E [ (\exp (X_t) - \exp (k_t))^+]   \\
&=&  E [ (\exp( X_t) - \exp (k_t))^+ 1_{\{ \nu_t X_t \in (x,y]\}}]
 +
  E [  (\exp( X_t) - \exp (k_t))^+ 1_{\{\nu_t  X_t > y\}}]  \\
  &\leq& (e^{y/\nu_t}-e^{x/\nu_t}) \P ( \nu_t X_t > x ) + E [ \exp( X_t)^p ]^{1/p} \P( \nu_t X_t > y)^{1/q}
   \\
  &\leq& (e^{y}-e^{x}) \P (\nu_tX_t > x ) + E [ \exp( X_t)^p ]^{1/p} \P(\nu_t X_t > y)^{1/q}
\end{eqnarray*}
where we have used H\"older's inequality and the existence of $p>1,t>0$ such that $E [ S_t^p ]<\infty$. Moreover, $E [ S_t^p ]$ is uniformly bounded as $t\to 0$, using Doob's maximal inequality for the martingale $S$. Now from LDP \eqref{eq:ldp:st} it follows
\[
 \limsup_{t \to 0}  t^{2H}L^2(t) \log \left(E [ (S_t - \exp (k_t))^+]
\right) \leq \max \left( - \Lambda(x) , - \frac{\Lambda(y)}{q}\right)
 \]
and we conclude by taking $y$ large enough (here we also use the goodness of the rate function, which implies that $\Lambda(y) \to \infty$ as $y \to \infty$.)

Now let us look at the lower bound. We have
\begin{eqnarray*}
E [ (S_t - \exp (k_t))^+]
&\geq&    E [ (\exp( X_t) - \exp (k_t)) 1_{\{\nu_{t} X_t > y\}}] \\
&\geq&     (\exp( y / \nu_t) - \exp (x/\nu_t))P( \nu_{t}X_t > y )\\
&\geq&     \exp (k_{t}) (\exp( (y-x) / \nu_t) - 1)P( \nu_{t}X_t > y ) \\
&\geq&    \frac{y-x}{\nu_t}  \exp (k_{t}) P( \nu_{t}X_t > y).
\end{eqnarray*}
Therefore
\begin{eqnarray*}
 t^{2H}L(t)^{2} \log E [ (S_t - \exp (k_t))^+] 
&\geq&     t^{2H}L(t)^{2} (k_{t}+ \log (y-x) -  \log \nu_t )+t^{2H}L(t)^{2} \log P( \nu_{t} X_t > y)
\end{eqnarray*}
and the first summand goes to $0$ as $t\to 0$. Therefore, for any $y>x$,
\[
\liminf_{{t\to 0}}
t^{2H}L(t)^{2} \log E [ (S_t - \exp (k_t))^+]
\geq
\liminf_{{t\to 0}}
t^{2H} L(t)^{2} \log P( \nu_{t}X_t > y)\geq - \Lambda(y).
\]
By continuity of $\Lambda$  \cite[Corollary 4.10]{FZ17} and the fact that the rate function is the same as for the self-similar process, this holds 
for $\Lambda(x)$ as well and the lower bound is proved.
\cvd
\end{proof}

The following implied volatility asymptotics is a consequence of the previous result and an application of \cite{gaolee}. Let us denote with $\sim$ asymptotic equivalence ($f_t\sim g_t$ iff $f_t/ g_t\to 1$).
\begin{corollary}
For model \eqref{eq:price-SDE}, 
let us assume that {\bf(A1)\rm}, {\bf(K1)\rm}, {\bf (K2)\rm},  {\bf($\bf\Sigma 1$)}, {\bf($\bf\Sigma 2$)} hold, that $S$ is a martingale and there exist $p>1,t>0$ such that $E [ S_t^p ]<\infty$. Then, with log-moneyness as in \eqref{eq:log:moneyness} and $x\neq 0$,  the short-time asymptotics for  implied volatility
\begin{equation}
\label{eq:ivol:as}
\sigma_{BS}(t,k_t)\to \frac{x}{\sqrt{2\Lambda(x)}}=:\Sigma_{LM}(x) \mbox{ as }  t\to 0
\end{equation}
 holds.
As a consequence, with $k_t'=x' t^{-H+1/2}L(t)^{-1}$, the finite difference implied volatility skew satisfies
\begin{equation}
\label{eq:skew:as}
\frac{\sigma_{BS}(t,k_t)-\sigma_{BS}(t,k'_t)}
{k_t-k'_t}
\sim\frac{\Sigma_{LM}(x)-\Sigma_{LM}(x')}{x-x'}
t^{H-1/2}L(t)
\end{equation}
\end{corollary}

\begin{remark}\rm\label{rem:lmfbm}

When taking the kernel in \eqref{eq:kernel:logmod}
with $0< H \leq 1/2$ we have
\begin{equation*}\label{eq:Lt}
L(t)\sim  (-\log t)^{-p}
\end{equation*}
in \eqref{eq:skew:as}, and the finite difference skew at the LDP regime explodes as $t^{H-1/2}(-\log t)^{-p}$. We prove this for $H>0$, because {(\bf K2)\rm} fails for $H=0$. However, even for $H=0$ the process is defined and the skew asymptotics \eqref{eq:skew:as} can be computed and is consistent with the  ``Gaussian'' result at the Edgeworth regime in \cite{BHP20}.
It is also clear that
\[
\frac{\Sigma_{LM}(x)-\Sigma_{LM}(x')}{x-x'}
\]
is an approximation of $\partial_{x}\Sigma_{{LM}}(0)$ for $x,x'$ close to $0$. Assuming $\Lambda$ smooth and, as one expects, $\Lambda(0)=0$ and $\Lambda'(0)=0$, we have
\[
\frac{x}{\sqrt{2\Lambda(x)}}=\frac{1}{\sqrt{\Lambda''(0)+\frac{\Lambda'''(0)x}{3}+O(x^2)}}=
\frac{1}{\sqrt{\Lambda''(0)}}\left(1-\frac{\Lambda'''(0)}{6\Lambda''(0)}x+O(x^2)\right)
\]
so that we can approximate the implied skew as
\[
\frac{\sigma_{BS}(t,k_t)-\sigma_{BS}(t,k'_t)}
{k_t-k'_t}
\approx
 \Sigma_{{LM}}'(0)
t^{H-1/2}L(t)
 = -\frac{\Lambda'''(0)}{6\Lambda''(0)^{3/2}}  t^{H-1/2}L(t).
\]
Note, however, that \eqref{eq:skew:as} and the asymptotics in \cite{BHP20} are different mathematical results. In addition, besides providing the at-the-money behaviour, result \eqref{eq:ivol:as} can also be used to compute the whole short-dated smile, including the wings, so it can be used for calibration and, for example, for tail risk hedging.
Since, as noted at the end of Section \ref{sec:modulated}, the rate function is the same as for the self-similar process and does not depend on the modulating function $L$, it can be computed as explained in Section \ref{sec:numerics} for fOU.

\end{remark}
\begin{proof} We first prove equation \eqref{eq:ivol:as}. Apply  Corollary 7.1 - Equation (7.2)  in \cite{gaolee}, along the lines of Appendix D in \cite{FGP21} or  Corollary 4.15 in \cite{FZ17}.
Then
\[
\sigma_{BS}^2(t,k_t)\sim -\frac{1}{t}\frac{k_t^2}{2\log c(t,k_t)}
\sim \frac{x^2}{2\Lambda(x)}
\]
and taking the square root we get the result. Equation
\eqref{eq:skew:as} follows easily from equation
\eqref{eq:ivol:as}.
\cvd
\end{proof}

\subsection{Large deviation pricing under fractional Ornstein-Uhlenbeck volatility }
As consequence of  Theorem \ref{th:LDP-log-price} and Theorem \ref{th:small-time-LDP} and the computations in Section \ref{sec:fOU}, we can derive asymptotic pricing formulas for European put and call options under the price dynamics in \eqref{eq:price-SDE}, with volatility driven by the process given in \eqref{eq:FOU}. In this case, we are considering the stochastic volatility dynamics
\begin{equation}\label{eq:stoch:vol:fOU}
\begin{cases}
dS_t &=S_t\sigma(V_t)d(\bar{\rho}\bar B_t+\rho B_t),\\
d V_{t}& = -a V_{t } dt + dB^{H}_{t}, 
\end{cases}
\end{equation}
with $S_{0}=1,\,V_{0}=0$.
Notice that this is written in differential form but $V$ could also be written explicitly as in \eqref{eq:FOU}.
With the same arguments used in the proof of Theorem \ref{th:ivol:logmod}, we have
$$
\P(S_{t}>  e^{x t^{1/2-H}})=\P(X_{t}> x t^{1/2-H})
\approx \exp\{-t^{-2H} J(x) \},
$$
where  $J(x)=I_{X_1}(x)$. More explicitly, \eqref{eq:rate-fun-finite} reads
\begin{equation}\label{J}
J(x)=\inf_{f\in H_0^1[0,1]} \left[\frac12 \lVert f\rVert_{H_0^1[0,1]}^2+\frac12 \frac{\big(x-\int_0^1 \rho \sigma(\hat{f}(t))\dot{f}(t)\,dt\big)^2}{\int_0^1 \bar{\rho}^2\sigma^2(\hat{f}(t)) dt}\right].
\end{equation}

 We have the following theorem.
\begin{theorem}\label{thm:main:smalltime}
Suppose {\bf($\bf\Sigma 1$)}, {\bf($\bf\Sigma 2$)} hold.
If $x<0$ and $k_t=x t^{{H-1/2}}$, the put price in short-time satisfies
\[
p(t,k_t)=E[(e^{k_t}-S_t)^+]\approx \exp\{-t^{-2H}  J(x) \}.
\]
In addition, we now assume that the process $S$ is a martingale and there exist $p>1,t>0$ such that $E [ S_t^p ]<\infty$.
If $x>0$ and $k_t=x t^{{H-1/2}}$, we have
\[
c(t,k_t)=E[(S_t-e^{k_t})^+]\approx \exp\{-t^{-2H} J(x) \}.
\]
\end{theorem}

\begin{remark}\label{rm:mart}
In both Theorems \ref{th:ivol:logmod} and \ref{thm:main:smalltime} the call price asymptotics holds under the assumption that the price process $S$ is a martingale, along with a moment condition. In the diffusive case ($H=1/2$) several related results are available. In particular, martingality holds if $\sigma$ has exponential growth and $\rho <0$ \citep{sin1998complications, jourdain2004loss, lions2007correlations}. Note that the assumption of negative correlation is justified from a financial perspective.

In the rough case, martingality is known to hold when $\sigma$ has linear growth and the driving process is the fBM \citep{FZ17}. In \cite{gassiat2018martingale}, it is 
shown that for a class of rough volatility models with $\sigma$ of exponential growth (that includes the rough 
Bergomi model) the stock price is a true martingale if and only if $\rho \leq 0$, while $E[S_{t}^{p}]=+\infty$ for $p>1/(1-\rho^{2})$, for any $t>0$. 

Models where the volatility is a function $\sigma$ of a Gaussian process are considered in \citep{Gu2}. If $\sigma$ grows faster than linearly, conditions for the explosion of moments are given both in the correlated and uncorrelated case.

For models \eqref{eq:Volterra-modulated} and \eqref{eq:stoch:vol:fOU}, these are open questions. We expect the conditions for the call asymptotics in Theorems \ref{th:ivol:logmod} and \ref{thm:main:smalltime} to hold in case $\rho <0$ and $\sigma$ with exponential growth. In particular, martingality should definitely hold in the cases analogous to \citep{gassiat2018martingale}, but with fOU driver. Indeed, the distribution of the fOU process is more concentrated than the one of the fBM, because of the mean reversion property.
\end{remark}

\begin{proof} This follows from the classic argument that we spelled out in the proof of Theorem \ref{th:ivol:logmod}. The proof follows as in Appendix C, Proof of Corollary 4.13 in \cite{FZ17}.
\cvd
\end{proof}
Again, from this call and put price asymptotics, an application of Corollary 7.1 in \cite{gaolee}  gives the following result.

\begin{corollary}
\label{corollary:ivol} 
Under the assumptions of Theorem \ref{thm:main:smalltime}, writing $k_t=xt^{1/2-H}$, we have, for $x\in \R \setminus \{0\}$, 
\begin{equation}\label{ivol:expansion}
\sigma_{BS}(t, k_t)
 \to\frac{|x|}{\sqrt{2J(x)}}=:
\Sigma_{fOU}(x), \mbox { as } t\to 0 
\end{equation}
\end{corollary}

As a consequence, the behavior of the implied skew at the large deviations regime under fOU-driven volatility is as follows.
\begin{corollary}
\label{corollary:skew} 
Under the assumptions of Theorem \ref{thm:main:smalltime}, writing $k_t=xt^{1/2-H}$, we have, for $x>0$, 
\begin{equation}\label{skew:expansion}
\frac{\sigma_{BS}(t, k_t)-\sigma_{BS}(t, -k_t)}{2 k_{t}}
 \sim \frac{\Sigma_{fOU}(x)-\Sigma_{fOU}(-x)}{2x} t^{H-1/2}, \mbox { as } t\to 0. 
\end{equation}
\end{corollary}

\begin{remark}[On moderate deviations]\label{th:moderate_deviations}

Model \eqref{eq:stoch:vol:fOU} should satisfy a moderate deviation result analogous to the ones in \cite{BFGHS} and Theorem 3.13 in \cite{FGP22}. 
Let $c(\cdot,\cdot)$ be as in \eqref{eq:putcall}, the price process $S$ given in \eqref{eq:stoch:vol:fOU}.
Assume that $J$ is $n\in \N$ times continuously differentiable.
Let $H\in (0,1/2)$, $\beta>0$ and $n\in \N$ such that 
$ \beta \in (\frac{2H}{n+1}, \frac{2H}{n}]$.
Set $\ell_t=x t^{1/2-H+\beta}$. Then, we can formally compute the call asymptotics from 
Theorem \ref{thm:main:smalltime}, plugging $\ell_{t}$ as log price instead of $k_{t}$, so that we substitute $x_{t}=xt^{\beta}$ to $x$ in a Taylor expansion of $J$ at $0$ and get
\[
J(x_t)=\sum_{i=2}^{n} \frac{J^{(i)}(0)}{i!} x^i t^{i\beta} +O(t^{(n+1)\beta}).
\]
Now, consider the speed $t^{{2H}}$ in Theorem \ref{thm:main:smalltime} and that $t^{(n+1)\beta-2H}\rightarrow 0$ if $ \beta \in(\frac{2H}{n+1}, \frac{2H}{n}]$, 
recall from  \cite{FZ17} and \cite{BFGHS} that $J(0)=J'(0)=0,\,J''(0)=1/\sigma(0)^{2}$ and we find that the call price should satisfy the following moderate deviations asymptotics, as $t\to 0$,
\[
\log 
c(t,\ell_t)= -
\sum_{i=2}^{n} \frac{J^{(i)}(0)}{i!} x^i t^{i\beta-2H}+O(t^{(n+1)\beta-2H}).
\]
We expect that a complete proof of this fact could be adapted from Proof of Theorem 3.13 in \cite{FGP22} or Proof of Theorem 3.4 in \cite{BFGHS}.
Assuming this call price asymptotics holds true, the following implied volatility asymptotics can be derived using Corollary 7.1, Equation (7.2) in \cite{gaolee} and that $J''(0)=1/\sigma(0)^{2}$
\begin{equation}\label{exp:imp:vol}
\begin{split}
\sigma_{BS}^2(t, \ell_t)=
\sum_{j=0}^{n-2}
(-1)^j 2^j \sigma(0)^{2(j+1)}
 \left(
\sum_{i=3}^{n} \frac{J^{(i)}(0)}{ i!} x^{i-2} t^{(i-2)\beta}
 \right)^j
+O(t^{2H-2\beta}).
\end{split}
\end{equation}
\end{remark}

\begin{remark}[On related results]
\label{rm:related} 
 A pathwise small-noise LDP under fOU volatility has been proved in \cite{horvath_jacquier_lacombe_2019}, with different hypothesis in particular on the function $\sigma$. From this LDP, a short-time result for a suitably renormalized process is also derived, with a time-scaling different from ours. 
 
In \cite{jacquier2020volterra} asymptotic results are given for Volterra driven volatility models, including large and moderate deviations, also in small-time. Hypothesis on the models are different from ours, for example $\sigma^{2}(x)$ is of linear growth, or alternatively a moment condition of type $\E[\sigma(V)^{2p}]<\infty$ for any $p\geq 1$ holds. The rate function  is given as an expression involving fractional derivatives of the minimiser.  In particular, in \cite[Section 4.2.1]{jacquier2020volterra} these results are applied to the rough Stein-Stein model, which is similar to \eqref{eq:stoch:vol:fOU}, with the RLp instead of the fBM, and with the specific choice of volatility function $\sigma^{2}(x)=x^{2}$. Analogous results should also hold with the fBM instead of the RLp as driver of the volatility.
 \end{remark}

\begin{remark}[On applications]
As mentioned in the introduction, short-time asymptotic approximations to the implied volatility surface are used for model calibration, pricing and other applications. They give information on option prices with short maturity, with low computational burden. This helps for example in the creation of delta-hedging strategies that are sensitive to short-term moves in the underlying and in general in trading and risk management.
Efficient and accurate methods for calibrating fOU-driven volatility models are relevant, for example, because these volatility models are used for computing option prices and implied volatilities \citep{GS17,GS20} and for hedging \citep{GS20hedging}.  Furthermore, \cite{GS18} compare the price impact of fast mean-reverting Markov stochastic volatility models with the price impact of mean reverting rough volatility models (see also \cite{GS19}). In \cite{FH18}, a model with both return and volatility driven by a fast mean reverting fOU process are used for portfolio optimization, in the $H>1/2$ regime.  
\end{remark}

\section{Numerical experiments}\label{sec:numerics}

In this section we test the accuracy of short-time pricing formulas \eqref{ivol:expansion}, and \eqref{exp:imp:vol} and of the implied skew asymptotics \eqref{skew:expansion}. We do so for a stochastic volatility model with asset price dynamics given by \eqref{eq:price-SDE}, with both fBM-driven volatility (i.e., $V=B^H$ is the fBM)
and fOU-driven volatility (i.e., $V=V^H$  is the fOU process, as in \eqref{eq:stoch:vol:fOU}).
Recall that both fBM and fOU models lead to the same rate function. 

For numerical experiments with log-fBM volatility, we refer to \cite{BHP20}. In particular, the discussion in Remark \ref{rem:lmfbm} on the at-the-money implied skew for log-modulated models is consistent with the numerical evaluations of at-the-money skews in \cite[Section 7]{BHP20}.

From Section \ref{sec:fOU}, we have the Volterra representation of the fBM 
$$
B^H_r = \int_0^t K_H(r,s) dB_s,
$$
where $K_H$ is the kernel in \eqref{eqn:kernelfbm}, and the Volterra representation of the fOU process
\begin{equation}\label{FOU}
V^H_r= \int_0^r K(r,s) ds = \int_0^r \Big(K_H(r,s)-a\int_s^r e^{-a(r-u)} K_H(u,s)du       \Big) dB_s.
\end{equation}
To evaluate the quality of approximations \eqref{ivol:expansion}, \eqref{skew:expansion} and \eqref{exp:imp:vol}, we first simulate Monte Carlo call prices under both these models, from which we then recover Black-Scholes implied volatilities. In both cases we consider a volatility function $\sigma(\cdot)$, depending on positive parameters $\sigma_{0}, \eta$ given by
\begin{equation}\label{volFun}
\sigma(x)=\sigma_0 \exp\Big({\frac{\eta}{2} \,x}\Big).
\end{equation}
To compute these prices under our stochastic volatility dynamics, we need to simulate the asset price at the fixed time horizon $t>0$. Hence we consider a time-grid $t_k= k\frac{t}{N}$, $k=0,\ldots, N$, and on this grid the random vector $(V_{t_1},\ldots, V_{t_N}, B_{t_1}, \ldots, B_{t_N})$, first with $V=B^H$ and then $V=V^H$. In both cases, it is a multivariate Gaussian vector with zero mean and known covariance matrix, that can be computed from the Volterra representation of the processes. The whole vector can be simulated using a Cholesky factorization of this covariance matrix. We then use this vector to construct an approximate sample of the log-asset price 
\[
X_t=-\frac 12\int_0^t\sigma^2(V_s) ds  + \rho\int_0^t \sigma(V_s) dB_s +\bar{\rho}\int_0^t\sigma(V_s) d\bar B_s
\]
by using a forward Euler scheme on the same time-grid
$$
X_t^N= -\frac{t}{2N}\sum_{k=0}^{N-1} \sigma^2(V_{t_k})+\sum_{k=0}^{N-1} \sigma(V_{t_k})\Big(\rho(B_{t_{k+1}}-B_{t_k} )+ \bar \rho(\bar B_{t_{k+1}}-B_{t_k}) \Big).
$$
We produce $M$ i.i.d. approximate Monte Carlo samples $(X_t^{N,m}, V_T^{m})_{1\leq m\leq M}$, that we use to evaluate call option prices by standard sample average. Then, we compute the corresponding implied volatilities $\sigma_{BS}(t, k)$ by Brent's method (see \cite{At08}, \cite{Pr08}), where $t$ is the maturity and $k$ the log-moneyness.

Note that Theorem \ref{thm:main:smalltime}, Corollary \ref{corollary:ivol} and  Corollary \ref{corollary:skew} do not apply to the model above, because
$\sigma(\cdot)$ does not satisfy the polynomial growth condition {\bf($\bf\Sigma 2$)}. However, also in in the self-similar case, large deviations pricing results were first obtained under  linear growth conditions in \cite{FZ17} and then the conditions were weakened  in \cite{bayer_regularity, Gu1} to include exponential growth. Therefore, we chose here to test our result on the exponential volatility in \eqref{volFun}, for which our result should hold as well. This choice is more realistic, being analogous to the rough Bergomi model, and being the volatility function considered e.g. in \cite{GS20hedging}.

To evaluate the accuracy of large deviations approximation \eqref{ivol:expansion}, we follow the choice in \cite{FGP22} and use as model parameters $H=0.3, \rho=-0.7, \sigma_0=0.2, \eta=1.5$, and as mean reversion parameter in fOU we take  $\alpha=1$ or  $\alpha=2$. These parameters are similar to the ones estimated on empirical volatility surfaces, as for example in \cite{bayer2016pricing}. We take a rough, but not ``extremely rough'' ($0.3$ instead of $0.1$) Hurst parameter $H$, motivated by the recent study \cite{GuyonSkew}.

We simulate $M=10^6$ Monte Carlo samples using $N=500$ discretization points. 
We estimate call option prices $E[(S_0e^{X_t}- S_0e^{k_t})^+ ]$, where $k_t= x t^{1/2-H}$, and the corresponding implied volatility $\sigma_{BS}(t, k_t)$.

Then, we need to compute $\Sigma_{{fOU}}$. The rate function $J$ in \eqref{J} can be approximated numerically using the Ritz method, as described in detail in \cite[Section 40]{gelfand_fomin}, \cite{FZ17}, and
\cite[Remark 4.3 and Section 5.1]{FGP22}. The rate function is obtained trough numerical optimization on a fixed, finite number of coefficients associated to a basis of the Cameron-Martin space $H_0^1$. We take as basis the Fourier basis, i.e. $\{\dot e_i\}_{i\in \N}$ with
$$
\dot e_1(s)= 1, \qquad \dot e_{2n}(s)=\sqrt{2} \cos(2\pi n s),\qquad \dot e_{2n+1}=\sqrt{2}\sin(2\pi n s), \quad n \in \N\setminus \{0\},
$$
that we truncate to $N=5$ (larger values of $N$ did not seem to improve the computation) and use the more explicit representation of the rate function $J$ in \eqref{J} given in
\cite{FZ17}, \cite[Proposition 5.1]{BFGHS}, \cite[Section 5.1]{FGP22}.

In Figure \ref{fig:1} we show for each model how, as the maturity $t$ becomes smaller, $\sigma_{BS}(t, k_t)$ gets closer to the asymptotic limit in \eqref{ivol:expansion}, where $k_t=x t^{1/2-H}$. We recall that $t$ is the option maturity and we numerically evaluate $\sigma_{BS}(t, k_t)$ for $t\in \{0.05, 0.1, 0.2, 0.3, 0.5 \}$ and $x \in [-0.2, 0.2]$ for $50$ equidistant points. The fact that, even for very small maturities, the short-time limit is not reached, can be explained by the fact that the error is of order $t^{2H}$ (as shown in the self-similar case in \cite{FGP22}), which vanishes as $t\to 0$, albeit slowly, since $H<1/2$.

In Figure \ref{fig:2}, for each fixed maturity, we compare the implied volatility smiles produced by each model (fOU vs fBM-driven volatilities), in order to observe the influence of the magnitude of the mean reversion parameter $a$ on the volatility smiles. In particular, we note that implied volatilities generated by fOU-driven models seem to fall between the implied volatilities generated by fBM-driven models and the asymptotic smile, indicating convergence also if polynomial growth of $\sigma(\cdot)$ is not satisfied in this example. 
\begin{figure}[t]\centering
\includegraphics[width=0.49\linewidth]{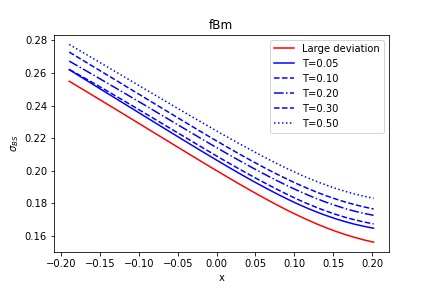}\\
\includegraphics[width=0.49\linewidth]{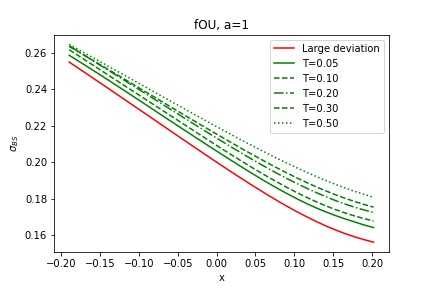}
\includegraphics[width=0.49\linewidth]{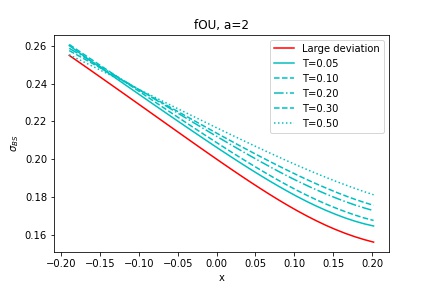}
\caption{Implied volatility smile with fBM-driven stochastic volatility, fOU-driven stochastic volatility with $a=1$, fOU-driven stochastic volatility with $a=2$ and large deviation approximation \eqref{ivol:expansion}. Model parameters: $H=0.3, \rho=-0.7, \sigma_0=0.2$ and $\eta=1.5$. Monte Carlo parameters: $10^6$ trajectories and $500$ time-steps. Recall that $k_t=x t^{1/2-H}$. The rate function is computed with the Ritz method with $N=5$ Fourier basis function.  \label{fig:1}
 }

\end{figure}

\begin{figure}[t]

\centering
\includegraphics[width=0.49\linewidth]{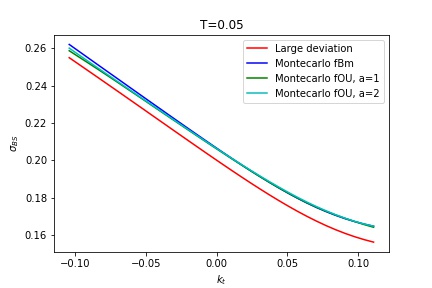}
\includegraphics[width=0.49\linewidth]{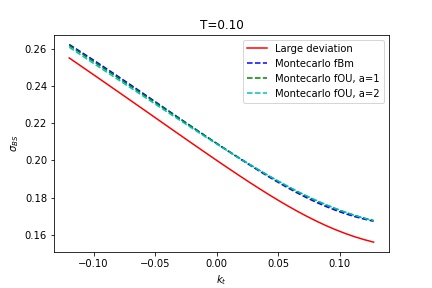}

\includegraphics[width=0.49\linewidth]{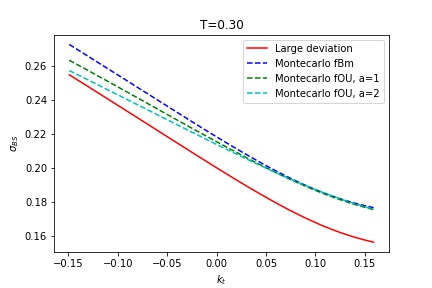}
\includegraphics[width=0.49\linewidth]{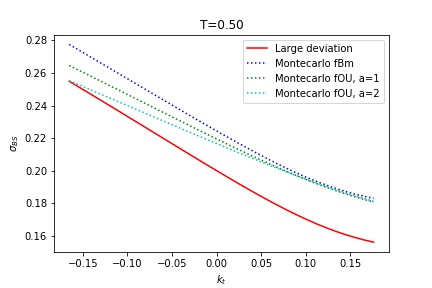}

\caption{Implied volatility smile with fBM-driven stochastic volatility, fOU-driven stochastic volatility with $a=1$, fOU-driven stochastic volatility with $a=2$ and large deviation approximation \eqref{ivol:expansion}. Model parameters: $H=0.3,\, \rho=-0.7, \sigma_0=0.2$ and $\eta=1.5$. Monte Carlo parameters: $10^6$ trajectories and $500$ time-steps. We plot each model fixing the time horizon and varying $x$, where $k_t=xt^{1/2-H}$. Rate function is computed with the Ritz method with $N=5$ Fourier basis function.
\label{fig:2}}
\end{figure}


\medskip
We test now the moderate deviation asymptotics in Remark \ref{th:moderate_deviations}. In order to do so, let us recall an expansion to the fourth order of the rate function that allows us to use the second order moderate deviation.\footnote{This expansion is given in \cite[Lemma 6.1]{FGP22}, where the kernel $C(t-s)^{H-1/2}$ is used. However, in the proof of this result the specific shape of the kernel is not used, but only self-similarity, and therefore it holds for $K(t,s)$ in \eqref{eqn:kernelfbm} as well.} We denote now $K_{H}f(t)=\int_0^t K_{H}(t,s)f(s)ds$ and with $\overline{K_{H}}$ the adjoint of $K_{H}$ in $L^2[0,1]$, so that $\overline{K_{H}}f(u)=\int_u^1 K_{H}(t,u) f(t) dt$, where again $K_{H}$ is the fBM kernel in \eqref{eqn:kernelfbm}.

\begin{lemma}[Fourth order energy expansion]
\label{expansion:J}
Let us assume that $\sigma(\cdot)$ is countinuously differentiable two times.
Let $J(x)$ be the energy function in \eqref{J}. Then
\begin{equation}\label{expanJ}
\begin{split}
J(x)&= \frac{J''(0)}{2} x^2 +
\frac{J'''(0)}{3!} x^3+\frac{J^{(4)}(0)}{4!} x^4+O(x^5)\,\\
\end{split}
\end{equation}
where
$$
J''(0)=\frac{1}{\sigma(0)^2}
,\quad J'''(0)=- 6 \frac{ \rho \sigma'(0)}{\sigma(0)^4} \langle K_{H}1,1 \rangle, 
$$
and
\[
\begin{split}
J^{(4)}(0)
 & =12  \frac{\sigma'(0)^2}{\sigma(0)^6}
\left\{
9 \rho^2 
 \langle K_{H}1,1\rangle^2
- \rho^2
 \langle (K_{H}1)^2 ,1\rangle  
-
 \langle (\overline{K_{H}}1)^2 ,1\rangle  
-
2 \rho^2 
\langle K_{H}1,\overline{K_{H}}1\rangle
\right\}+\\
& -12 \frac{\sigma''(0)}{\sigma(0)^5}   \rho^2
 \langle (K_{H}1)^2 ,1\rangle. \end{split}
\]

Plugging \eqref{expanJ} into \eqref{exp:imp:vol} and fixing $n=4$, from straightforward computations, we obtain the equivalent asymptotic formula
\begin{equation}\label{equiFor}
\sigma(t,\ell_t)=\Sigma_{fOU}(0)+\Sigma_{fOU}'(0)x t^\beta + \frac{\Sigma_{fOU}''(0)}{2}x^2 t^{2\beta}+o(t^{2H-2\beta})
\end{equation}
where
\[
\begin{split}
\Sigma_{fOU}(0)&= \sigma(0), \quad
\Sigma_{fOU}'(0)=\frac{\rho \sigma'(0)\langle K_{H}1,1\rangle}{\sigma(0)}, \\
\frac{\Sigma_{fOU}''(0)}{2}
&=
 \frac{\sigma'(0)^2}{\sigma(0)^3}
\left\{
- 3 \rho^2 
 \langle K_{H}1,1\rangle^2
+ \frac{\rho^2}{2}
 \langle (K_{H}1)^2 ,1\rangle  
+
\frac{ 1}{2}
 \langle (\overline{K_{H}}1)^2 ,1\rangle  
+
 \rho^2 
\langle K1,\overline{K_{H}}1\rangle
\right\}\\
&+ \frac{\sigma''(0)}{\sigma(0)^2}  \frac{ \rho^2 }{2}
 \langle (K_{H}1)^2 ,1\rangle.  
\end{split}
\]
\end{lemma}
We plot in Figure \ref{fig:3} implied volatilities computed via Monte Carlo simulations and the corresponding approximation given in \eqref{equiFor}.
We take again $\sigma(\cdot)$ as in\eqref{volFun}, with parameters $H=0.3, \rho=-0.7, \sigma_0=0.2, \eta=0.2$, and $\beta=0.125$. We first note that we fix $n=4$ in \eqref{exp:imp:vol}, and so we choose $\beta \in (\frac{2H}{n+1}, \frac{2H}{n}]$, that is the interval $(0.12, 0.15]$. 
With respect to our previous experiments, we also take the smaller vol of vol parameter $\eta=0.2$, which is in line with the choices in \cite{BFGHS,FGP22}. Indeed, the quality of the approximation deteriorates as $\eta$ grows, and for larger $\eta$ the asymptotic formula \eqref{equiFor} is accurate on a smaller time interval.   

\begin{figure}[t]
\centering
\includegraphics[width=0.7\linewidth]{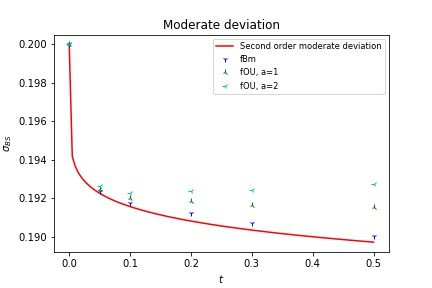}
\caption{Moderate deviation implied volatilities with fBM-driven stochastic volatility (blue), with fOU-driven stochastic volatility with $a=1$ (green) and fOU-driven stochastic volatility with $a=2$ (light blue), $\ell_t=x t^{1/2-H+\beta}$ with $x=0.3$ and $\beta=0.125$. Model parameters: $H=0.3, \rho=-0.7, \sigma_0=0.2$ and $\eta=0.2$. Monte Carlo parameters: $10^7$ trajectories, $500$ time-steps. We plot each model with fixed $x$ and varying maturity.
\label{fig:3}}
\end{figure}

In Figure \ref{fig:skew} we compare, with $k_{t}=x t^{H-1/2}$, $x>0$, the absolute value of the large deviations finite difference implied skew
\begin{equation}\label{eq:abs:skew}
\Psi_{t}:=\frac{|\sigma_{BS}(t, k_t)-\sigma_{BS}(t, -k_t)|}{2 k_{t}}
\end{equation}
computed on fBm-driven and fOU-driven stochastic volatility models, with the asymptotic skew expected from Corollary \ref{corollary:skew}, where we also use the approximation, as $x\to 0$,
\[
\frac{\Sigma_{fOU}(x)-\Sigma_{fOU}(-x)}{2x}
\sim
\Sigma_{fOU}'(0)=\frac{\rho \sigma'(0)\langle K_{H}1,1\rangle}{\sigma(0)}.
\]
We observe that, consistently with the smile slopes observed in Figure \ref{fig:2},
larger mean-reversion parameters $a$ correspond to flatter smiles and smaller skews (in absolute value), 
smaller mean-reversion parameters $a$ correspond to steeper smiles and larger skews (in absolute value),
fBm has a larger skew than fOU, and the asymptotic skew is even larger than the one generated from fBm, although very close to it. As maturity $t\to 0$, the difference between all these skews vanishes.

This could reflect the fact that larger mean-reversion parameters $a$ give more concentrated volatility trajectories, with $V$ in \eqref{FOU} staying closer to $0$ and therefore the stochastic volatility path $(\sigma_0 \exp(\frac{\eta}{2} \,V_{t}))_{t>0}$ staying closer to the spot-vol $\sigma_0$. This may produce flatter implied volatility surfaces and explain smiles and skews observed in Figures \ref{fig:2}, \ref{fig:3} and \ref{fig:skew} corresponding to larger $a$'s. On the short end of the surface, however, all of these have to coincide due to our asymptotic results.
Let us also mention that the discrepancy observed in Figure \ref{fig:2} on the level of the smile (regardless of the skew), between the asymptotic red line and all the simulated ``positive maturity'' lines is likely due to a term-structure term of order $t^{2H}$, for which we refer the reader to \cite{FGP22} (large deviations setting) and \cite{euch2018short} (central limit setting).

\begin{figure}[t]
\centering
\includegraphics[width=0.495\linewidth]{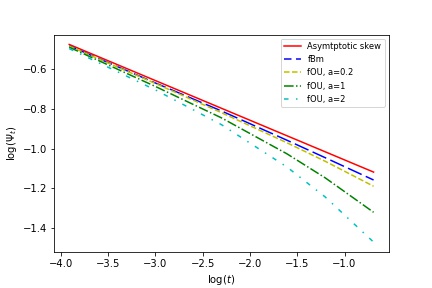}
\includegraphics[width=0.495\linewidth]{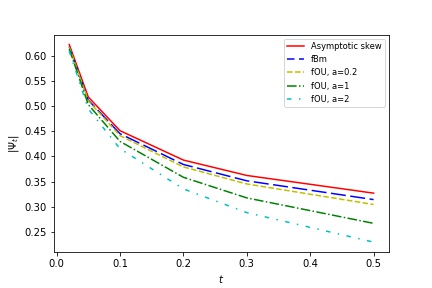}
\caption{ Implied skew in \eqref{eq:abs:skew} with fBm-driven stochastic volatility, fOU-driven stochastic volatility with $a=0.2$, $a=1$, $a=2$ and asymptotic skew from \eqref{skew:expansion}. Model parameters: $H=0.3,  \rho=-0.7, \sigma_0=0.2$ and $\eta=1.5$. Monte Carlo parameters: $10^7$ trajectories and $500$ time-steps. We take $x=0.01$ (recall that $k_t=x t^{1/2-H}$). Log-plot on the left hand side, linear plot on the right hand side.
\label{fig:skew}}
\end{figure}

\section{Conclusion}
In this paper, we prove a short-time large deviation principle for stochastic volatility models, with volatility given as a function of a Volterra process. This result holds without strict self-similarity assumptions on the processes driving the model, and can therefore be applied to some notable examples of (non self-similar) rough volatility models.

 We  first consider an application to the log-modulated rough stochastic volatility models introduced in \cite{BHP20}. We derive short-maturity asymptotics for European option prices and implied volatility surfaces. Our results on the implied skew at the large deviations regime are consistent with the results at the Edgeworth central-limit regime derived in \cite{BHP20}, but allow for valuation of options further from the money.
  
Then we consider models where volatility is given as a function of a fractional Ornstein-Uhlenbeck process, as e.g. in the seminal work \cite{GJR18}. In this case we find that the limit short maturity behavior of option prices and implied volatilities, as well as the short time implied skew, is the same as the one of a model driven by a fractional Brownian motion. We investigate this fact numerically on  simulation results, discussing also moderate deviations pricing and implied skew asymptotics.

\FloatBarrier

\begin{appendices}
\section{The large deviations principle}\label{app:ldp}
Large deviations give an asymptotic computation of small
probabilities on an exponential scale (see e.g.
\cite{DemZei} as a reference on this topic). We recall some
basic definitions (see e.g. Section 1.2 in \cite{DemZei}).
Throughout this paper a speed function is a sequence $\{v_n:n\geq
1\}$ such that $\lim_{n\to\infty}v_n=\infty$. A sequence of random
variables $\{Z_n:n\geq 1\}$, taking values on a topological space
$\mathcal{X}$, satisfies the large deviation principle (LDP) with rate function $I$ and speed function $v_n$ if
$I:\mathcal{X}\to[0,\infty]$ is a lower semicontinuous function,
$$\liminf_{n\to\infty}\frac{1}{v_n}\log P(Z_n\in O)\geq-\inf_{x\in O}I(x)$$
for all open sets $O$, and
$$\limsup_{n\to\infty}\frac{1}{v_n}\log P(Z_n\in C)\leq-\inf_{x\in C}I(x)$$
for all closed sets $C$. A rate function is said to be \emph{good} if all
its level sets $\{\{z\in\mathcal{Z}:I(z)\leq\eta\}:\eta\geq 0\}$
are compact.
Therefore,  if an LDP holds, and $\Gamma$ is a Borel set such that  $\inf_{x\in \Gamma^o}I(x)=\inf_{x\in \bar\Gamma}I(x)$ ($\Gamma^o$ and $\bar\Gamma$ are the interior and the  closure of $\Gamma$ respectively), then
$$\lim_{n\to\infty}\frac{1}{v_n}\log P(Z_n\in \Gamma)=-I(\Gamma)$$
where $ I(\Gamma)= \displaystyle\inf_{x\in \Gamma^o}I(x)=\inf_{x\in \bar\Gamma}I(x).$
 In this case we write
$$
P(Z_n\in \Gamma)\approx e^{- {I(\Gamma)}{v_n}}.
$$ 

Moreover $\{Z_n:n\geq 1\}$ is exponentially tight
with respect to the speed function $v_n$ if, for all $b>0$, there
exists a compact $K_b\subset\mathcal{X}$ such that
$$\limsup_{n\to\infty}\frac{1}{v_n}\log P(Z_n\notin K_b)\leq-b.$$

The concept of exponential tightness plays a crucial role in large
deviations; in fact this condition is often required to establish
that the LDP holds for a sequence of random
variables taking values on an infinite dimensional topological
space. In this paper we refer to condition (8) and (9) in Section 2 in   \cite{MacPac})
   which yield the exponential tightness when
the topological space $\mathcal{X}$ of the continuous function is equipped with the
topology of the uniform convergence.

\end{appendices}


\small

\begin{thebibliography}{00}

\bibitem[Ait-Sahalia et al.(2020)]{aitsahalia2019}
Y. Ait-Sahalia, C. Li, and C.X. Li (2020)
Implied stochastic volatility models. 
\emph{The Review of Financial Studies}, 34:394--450.

\bibitem[Al\`os and Leon(2017)]{alos2017curvature}
E.~Al{\`o}s and J.~A. Le{\'o}n (2017) 
On the curvature of the smile in stochastic volatility models. 
\emph{SIAM Journal on Financial Mathematics}, 8(1):373-399.

\bibitem[Al\`os et al.(2007)]{alos2007short}
E.~Al{\`o}s, J.~A. Le{\'o}n and J.~Vives (2007)
\newblock On the short-time behavior of the implied volatility for
  jump-diffusion models with stochastic volatility.
\newblock {\em Finance and Stochastics}, 11(4):571--589.

\bibitem[Atkinson(2008)]{At08}
 K.E. Atkinson (2008)
\textit{An introduction to numerical analysis, 2nd ed}, Wiley India Pvt. Limited

\bibitem[Baldi and Pacchiarotti(2022)]{BaPa}
P. Baldi and B. Pacchiarotti (2022) Large Deviations of continuous Gaussian processes:
from small noise to small time. {\em  Preprint arXiv:2207.12037}.

\bibitem[Bayer et al.(2020)]{bayer_regularity}
C. Bayer, P. K. Friz, P. Gassiat, J. Martin and B. Stemper (2020) 
\newblock A regularity structure for rough volatility. 
\newblock {\em Mathematical Finance}, 30(3):782--832.

\bibitem[Bayer et al.(2016)]{bayer2016pricing}
C.~Bayer, P.~K. Friz and J.~Gatheral (2016)
\newblock Pricing under rough volatility.
\newblock {\em Quantitative Finance}, 16(6):887--904.

\bibitem[Bayer et al.(2019)]{BFGHS}
C. Bayer, P. K. Friz, A. Gulisashvili, B. Horvath and B. Stemper (2019)  
Short-time near-the-money skew in rough fractional volatility models. 
{\em Quantitative Finance}, 19(5):779--798.

\bibitem[Bayer et al.(2021)]{BHP20}
C. Bayer, F. Harang and P. Pigato (2021) Log-modulated rough stochastic volatility models.
{\em SIAM Journal on Financial Mathematics}, 12(3):1257--1284.

\bibitem[Bourgey et al.(2023)]{BDFP2022}
F. Bourgey, S. De Marco, P. K. Friz and P. Pigato (2023) Local volatility under rough volatility.
{\em Mathematical Finance},  33(4):1119-1145

\bibitem[Catalini and Pacchiarotti(2023)]{CatPac}	
G. Catalini and B. Pacchiarotti (2023) Asymptotics for multifactor Volterra type stochastic volatility models. {\em Stochastic Analysis and Applications}, 41(6):1025-1055.


\bibitem[Cellupica and Pacchiarotti(2021)]{CelPac}	
M. Cellupica and B. Pacchiarotti (2021) Pathwise Asymptotics for Volterra Type Stochastic Volatility Models. {\em Journal of  Theoretical  Probability}, 34(2):682--727.

\bibitem[Cheridito et al.(2003)]{Ch-ka-Ma}
P. Cheridito, H.  Kawaguchi and M. Maejima (2003)
Fractional Ornstein-Uhlenbeck processes. {\em Electronic Journal of   Probability}, 8:1--14.

\bibitem[Dembo and Zeitouni(1998)]{DemZei}
A. Dembo and  O. Zeitouni (1998)  \textit{Large Deviations Techniques and Applications}, Jones and Bartlett, Boston MA.

\bibitem[El Euch et al.(2019)]{euch2018short}
O.~El~Euch, M.~Fukasawa, J.~Gatheral and M.~Rosenbaum (2019)
\newblock Short-term at-the-money asymptotics under stochastic volatility
  models.
\newblock {\em SIAM Journal on Financial Mathematics}, 10(2):491--511.

\bibitem[Forde et al.(2020)]{forde_bergomi_H0}
M.~Forde, M.~Fukasawa, S.~Gerhold and B.~Smith (2022)
\newblock The Riemann-Liouville field and its GMC as $H\to0$, and skew flattening for the rough Bergomi model.
\newblock {\em Statistics and  Probability  Letters},  181:10926.

\bibitem[Forde et al.(2021)]{forde_heston_H0}
M. Forde, S. Gerhold and B. Smith (2021)
Small-time, large-time and {$H\to 0$} asymptotics for the rough
  Heston model.
 {\em Mathematical Finance}, 31:203--241.

\bibitem[Forde and Zhang(2017)]{FZ17}
M. Forde and H. Zhang (2017) Asymptotics for rough stochastic volatility models.
\newblock {\em SIAM Journal on Financial Mathematics}, 8(1):114--145.

\bibitem[Friz et al.(2021)]{FGP21}
P.K. Friz, P. Gassiat and P. Pigato (2021)
Precise asymptotics: Robust stochastic volatility models.
{\em Annals of  Applied  Probability}, 31(2):896--940.

\bibitem[Friz et al.(2022)]{FGP22}
P.K. Friz, P. Gassiat and P. Pigato (2022)
Short dated smile under Rough Volatility: asymptotics and numerics.
{\em Quantitative Finance}, 22(3):463--480.

\bibitem[Friz et al.(2021)]{fps2020}
P.K. Friz, P.~Pigato and J.~Seibel (2021)
\newblock {The Step Stochastic Volatility Model}.
\newblock {\em Risk magazine},  June  2021, {\em (Longer version available at
at SSRN: https://ssrn.com/abstract=3595408)}.

\bibitem[Fukasawa(2011)]{fukasawa2011asymptotic}
M.~Fukasawa (2011)
\newblock Asymptotic analysis for stochastic volatility: Martingale expansion.
\newblock {\em Finance and Stochastics}, 15:635--654.

\bibitem[Fukasawa(2017)]{fukasawa2017short}
M.~Fukasawa (2017)
\newblock Short-time at-the-money skew and rough fractional volatility.
\newblock {\em Quantitative Finance}, 17(2):189--198.

\bibitem[Fukasawa(2020)]{fukasawa2020}
M.~Fukasawa (2020)
\newblock Volatility has to be rough.
 {\em Quantitative Finance}, 21(1):1--8.

 \bibitem[Gao and Lee(2014)]{gaolee}
K. Gao and R. Lee. (2014)
 Asymptotics of implied volatility to arbitrary order.
{\em Finance and Stochastics}, 18:349--392.

\bibitem[Garnier and S{\o}lna(2017)]{GS17}
J. Garnier and K. S{\o}lna (2017)
Correction to Black-Scholes Formula Due to Fractional Stochastic Volatility.
\emph{SIAM Journal on Financial Mathematics}, 8:560--588.


\bibitem[Garnier and S{\o}lna(2018a)]{GS18}
J. Garnier and K. S{\o}lna (2018a) 
Option pricing under fast-varying and rough stochastic volatility.
\emph{Annals of Finance}, 14:489--516.

\bibitem[Garnier and S{\o}lna(2019)]{GS19}
J. Garnier and K. S{\o}lna (2019) 
Option pricing under fast-varying long-memory stochastic volatility. 
\emph{Mathematical Finance}, 29:39--83.

\bibitem[Garnier and S{\o}lna(2020a)]{GS20}
J. Garnier and K. S{\o}lna (2020a) 
Implied Volatility Structure in Turbulent and Long-Memory Markets.
\emph{Frontiers in Applied Mathematics and Statistics }, 29 April 2020.

\bibitem[Garnier and S{\o}lna(2020b)]{GS20hedging}
J. Garnier and K. S{\o}lna (2020b)
Optimal hedging under fast-varying stochastic volatility.
\emph{SIAM Journal on Financial Mathematics}, 11(1):274--325

\bibitem[Gassiat(2019)]{gassiat2018martingale} 
P. Gassiat (2019) On the martingale property in the Rough Bergomi model. 
\emph{Electronic Communications in Probability}, 24:1--9


\bibitem[Gatheral et al.(2018)]{GJR18}
J. Gatheral, T.Jaisson and M.~Rosenbaum (2018)
 Volatility is rough.
{\em Quantitative Finance}, 18(6):933--949.


\bibitem[Gelfand and  Fomin(1963)]{gelfand_fomin}
I.M. Gelfand and S.V. Fomin (1963)
Calculus of variations.
Revised English edition translated and edited by Richard A. Silverman Prentice-Hall. 


\bibitem[Gulisashvili(2018)]{Gu1}
A. Gulisashvili (2018) Large Deviation Principle for Volterra type Fractional Stochastic Volatility Models. {\em SIAM Journal on Financial Mathematics},  9(3):1102--1136.

\bibitem[Gulisashvili(2020)]{Gu2}
A. Gulisashvili (2020)  Gaussian stochastic volatility models: scaling regimes, large deviations, and moment explosions. {\em Stochastic Processes and their  Applications}, 130(6):3648--3686.

\bibitem[Gulisashvili(2021)]{gulisashvili2020timeinhomogeneous}
A. Gulisashvili (2021)
Time-inhomogeneous Gaussian stochastic volatility models: Large
  deviations and super roughness.
{\em Stochastic Processes and their  Applications}, 139:37--79.

\bibitem[Gulisashvili(2022)]{gulisashvili2022multi}
A. Gulisashvili (2022)
Multivariate Stochastic Volatility Models and Large Deviation Principles.
{\em  Preprint arXiv:2203.09015}.


\bibitem[Gulisashvili et al.(2018a)]{GVZ}
A. Gulisashvili, F. Viens and X. Zhang (2018a) \textit{Small-Time Asymptotics for Gaussian Self-Similar Stochastic Volatility Models}, Applied Mathematics \& Optimization, 1--41.

\bibitem[Gulisashvili et al.(2018b)]{GVZ2}
A. Gulisashvili, F. Viens and X. Zhang (2018b) Extreme-strike asymptotics for general Gaussian stochastic volatility models. {\em Annals of Finance}, 15(1):59--101.
 
\bibitem[Guyon(2021)]{guyon} 
J. Guyon (2021) Dispersion-Constrained Martingale Schr\"odinger Problems and the Exact Joint S\&P 500/VIX 
Smile Calibration Puzzle. 
{\em Available 
at SSRN: https://ssrn.com/abstract=3853237 or http://dx.doi.org/10.2139/ssrn.3853237
}

\bibitem[El Amrani and Guyon(2023)]{GuyonSkew} 
M. El Amrani and J. Guyon (2023) Does the term-structure of equity at-the-money skew really follow a power law?
\newblock {\em Risk magazine},  July  2023, {\em (Longer version available at
at SSRN: https://ssrn.com/abstract=4174538)}.


\bibitem[Fouque and Hu(2018)]{FH18}
J.P. Fouque and  R. Hu (2018)
Optimal Portfolio under Fast Mean-Reverting Fractional Stochastic Environment 
\newblock \emph{SIAM Journal of Financial Mathematics}, 9(2):564-601.


\bibitem[Hager and Neuman(2021)]{hager2020multiplicative}
P. Hager and E. Neuman (2021)
 {The Multiplicative Chaos of $H=0$ Fractional Brownian Fields}.
{\em Annals of applied probability}, 32(3) 2139--2179.

\bibitem[Horvath et al.(2019)]{horvath_jacquier_lacombe_2019}
 B. Horvath, A. Jacquier and C. Lacombe (2019) Asymptotic behaviour of randomised fractional volatility models. {\em Journal of Applied Probability}, 56(2):496--523.

\bibitem[Jacquier et al.(2018)]{jacquier2017pathwise}
A. Jacquier, M.S. Pakkanen and H. Stone (2018)
Pathwise large deviations for the rough Bergomi model. 
\emph{Journal of Applied Probability}, 55(4):1078--1092.


\bibitem[Jacquier and Pannier(2022)]{jacquier2020volterra}
A. Jacquier and A. Pannier (2022)
Large and moderate deviations for stochastic Volterra systems,
{\em Stochastic Processes and their  Applications}, 
149:142--187.

\bibitem[Jourdain(2004)]{jourdain2004loss}
B. Jourdain (2004)
Loss of martingality in asset price models with lognormal stochastic volatility,
{\em preprint Cermics}, 
267:2004. 

\bibitem[Lee(2005)]{LEE}
R. W. Lee (2005)
Implied volatility: Statics, dynamics, and probabilistic interpretation, in Recent Advances in
Applied Probability, Springer, New York,  241-268.

\bibitem[Lions and Musiela(2007)]{lions2007correlations}
P. L. Lions and M. Musiela (2007)
Correlations and bounds for stochastic volatility models,
{\em In Annales de l'Institut Henri Poincare (C) Non Linear Analysis}, 
24:1--16.
 
\bibitem[Macci and  Pacchiarotti(2017)]{MacPac}
C. Macci and B. Pacchiarotti (2017) Exponential tightness for Gaussian processes with applications to some  sequences of weighted means.  {\em Stochastics}  89(2):469--484.

\bibitem[Medvedev and Scaillet(2003)]{MS03}
A. Medvedev and O. Scaillet (2003)
A simple calibration procedure of stochastic volatility models with jumps by short term asymptotics. 
\emph{Research Paper No. 93, September 2003, FAME  International Center for Financial Asset Management and Engineering. Available at SSRN 477441, 2003.}

\bibitem[Medvedev and Scaillet(2007)]{MS07}
A., Medvedev and O. Scaillet (2007)
Approximation and calibration of short-term implied volatilities under jump-diffusion stochastic volatility. 
\emph{The Review of Financial Studies}, 20(2):427--459.

\bibitem[Neuman and Rosenbaum(2018)]{NR18}
E. Neuman and M. Rosenbaum (2018)
Fractional {B}rownian motion with zero {H}urst parameter: a rough
  volatility viewpoint.
{\em Electronic Communications in Probability}, 23(61):1--12.


\bibitem[Nualart(2006)]{Nualart:06}
      D. Nualart (2006)
 \textit{ Malliavin Calculus and Related Topics},
     	Springer, Berlin.

\bibitem[Osajima(2015)]{osajima2015general}
Y. Osajima (2015)
 General asymptotics of Wiener functionals and application to implied volatilities. \emph{In Large Deviations and Asymptotic Methods in Finance},  137--173, Springer.


\bibitem[Pigato(2019)]{pig19}
P.~Pigato (2019)
\newblock { Extreme at-the-money skew in a local volatility model}.
\newblock {\em Finance and Stochastics}, 23:827--859.

\bibitem[Press et al.(2012)]{Pr08}
 W. H. Press, S. A. Teukolsky, W. T. Vetterling and  B. P. Flannery (2007) \textit{Numerical Recipes: The Art of Scientific Computing (3rd ed.)}. New York: Cambridge University Press. 

\bibitem[Sin(1998)]{sin1998complications}
C. A. Sin (1998)  \newblock {Complications with stochastic volatility models}. 
\newblock {\em Advances in Applied Probability}, 30(1):256--268.




\end{thebibliography}
\end{document}